\documentclass[11pt]{article}

\usepackage{amsmath}
\usepackage{amsthm}
\usepackage[inline]{enumitem}

\setitemize{topsep=4pt,parsep=4pt,partopsep=4pt}

\newcommand{\fullVersion}{}
\ifdefined\fullVersion
\else
\newcommand{\shortVersion}{}
\fi
\ifdefined\shortVersion
\usepackage[compact]{titlesec}
\titlespacing{\section}{3pt}{*2}{*2}
\titlespacing{\subsection}{2pt}{*1.5}{*1.5}
\titlespacing{\subsubsection}{1pt}{*1}{*1}
\allowdisplaybreaks
\newtheoremstyle{exampstyle}
  {.5em} 
  {.5em} 
  {} 
  {} 
  {\bfseries} 
  {.} 
  {.5em} 
  {} 
\theoremstyle{exampstyle} \newtheorem{theorem}{Theorem}
\theoremstyle{exampstyle} \newtheorem{corollary}{Corollary}
\theoremstyle{exampstyle} 
\theoremstyle{exampstyle} 
\theoremstyle{exampstyle} \newtheorem{definition}{Definition}
\theoremstyle{exampstyle} \newtheorem{lemma}{Lemma}
\fi

\usepackage[utf8]{inputenc}
\usepackage[T1]{fontenc}
\usepackage{lmodern}
\usepackage{color}
\usepackage{bm}

\usepackage[margin=1in]{geometry}

\usepackage{makeidx}
\usepackage{graphicx,comment}
\usepackage[linesnumbered,vlined]{algorithm2e}

\usepackage{multirow}
\usepackage{amsfonts, amssymb}
\usepackage{cite}
\usepackage{fullpage}
\usepackage{hhline}
\usepackage{bbm}
\usepackage{thmtools}
\usepackage{thm-restate}

\usepackage{pifont}
\newcommand{\cmark}{\ding{51}}%
\newcommand{\xmark}{\ding{55}}%

\ifdefined\fullVersion
\newtheorem{theorem}{Theorem}[section]
\newtheorem{lemma}[theorem]{Lemma}
\newtheorem{corollary}[theorem]{Corollary}

\newtheorem{definition}[theorem]{Definition}

\fi

\newenvironment{lemma-repeat}[1]{\begin{trivlist}
		\item[\hspace{\labelsep}{\bf\noindent Lemma \ref{#1} }]\em }%
	{\end{trivlist}}
\newenvironment{theorem-repeat}[1]{\begin{trivlist}
		\item[\hspace{\labelsep}{\bf\noindent Theorem \ref{#1} }]\em }%
	{\end{trivlist}}

\newcommand{\size}[1]{\ensuremath{\left|#1\right|}}
\newcommand{\set}[1]{\left\{ #1 \right\}}
\newcommand{\angles}[1]{\left\langle #1 \right\rangle}
\newcommand{\parentheses}[1]{\left( #1 \right)}
\newcommand{\brackets}[1]{\left[ #1 \right]}

\newcommand{\logp}[1]{\log\parentheses{#1}}

\newcommand*\samethanks[1][\value{footnote}]{\footnotemark[#1]}

\newcommand{\ceil}[1]{\lceil #1 \rceil}
\newcommand{\floor}[1]{\lfloor #1 \rfloor}

\newcommand{\eps}{\epsilon}

\newcommand{\promiseVerify}{{\sc Promise\&Verify}}

\newcommand{\OPT}{\ensuremath{\mathit{OPT}}}

\newcommand{\DLB}{\ensuremath{\mathit{DLB}}}
\newcommand{\MaxCut}{\ensuremath{\mathit{Max\text{-}Cut}}}
\newcommand{\MaxDiCut}{\ensuremath{\mathit{Max\text{-}DiCut}}}
\newcommand{\MVC}{\ensuremath{\mathit{MVC}}}
\newcommand{\MaxM}{\ensuremath{\mathit{MaxM}}}
\newcommand{\MaxIS}{\ensuremath{\mathit{MaxIS}}}

\newcommand{\MDS}{\ensuremath{\mathit{MDS}}}
\newcommand{\MEDS}{\ensuremath{\mathit{MEDS}}}
\newcommand{\MFVS}{\ensuremath{\mathit{MFVS}}}
\newcommand{\MFES}{\ensuremath{\mathit{MFES}}}
\renewcommand{\k}[1]{\ensuremath{k\text{-}}#1}
\newcommand{\p}[1]{\ensuremath{p_k\text{-}}#1}
\renewcommand{\v}[1]{\ensuremath{v_k\text{-}}#1}
\newcommand{\kMVC}{\k{\MVC}}
\newcommand{\kMaxM}{\k{\MaxM}}
\newcommand{\kMaxIS}{\k{\MaxIS}}

\newcommand{\ThmDiam}
{
There exists an $O(k)$ rounds deterministic algorithm in the CONGEST model that terminates with all vertices outputting \emph{SMALL} if the diameter is bounded by $k$, and \emph{LARGE} if the diameter is larger than $2k$. If the diameter is between $k+1$ and $2k$, the vertices answer unanimously, but may return either \emph{SMALL} or \emph{LARGE}.
}
\newcommand{\ThmMFVSLB}{
For any $k\in\mathbb N^+$, any algorithm that solves \MFVS{} or \MFES{} on a graph with a solution of size $k$ to within an additive error of $O(n/D)$ 
	must take $\Omega(D)$ \mbox{rounds in the LOCAL model.}
}
\newcommand{\ThmLBMain}{
For any $\epsilon=\Omega(1/n)$ and $\delta<1/2-EXP(-\Omega(n\epsilon))$, any Monte-Carlo LOCAL algorithm that computes a $(1+\epsilon)$-multiplicative approximation with probability at least $1-\delta$ for a problem $P \in \set{\MVC, \MaxM, \MaxIS, \MDS, \MEDS{}, \MaxDiCut{}, \MaxCut{}, \MFVS{}, \MFVS{}}$ 
requires $\Omega(\epsilon^{-1})$ rounds.
}

\newcommand{\ThmLBMainParam}{
There exists a family of graphs $G_{k,n}(V,E)$,
such that for any $\epsilon=\Omega(1/k)$ and $\delta<1/2-EXP(-\Omega(k\epsilon))$, any Monte-Carlo LOCAL algorithm that computes a $(1+\epsilon)$-multiplicative approximation with probability $1-\delta$ for some {\k{$P$}} $\in$ {\k{$\mathcal P$}} requires $\Omega(\epsilon^{-1})$ rounds. Here, $n = \size{V}$
can be arbitrarily larger~than~$k$.
}

\newcommand{\ThmCKPParam}{
There exists $\underline{k}\in\mathbb N$ such that for any 
    $\underline{k}\le k\le 0.99n$, there exists a family of graphs $G_{k,n}(V,E)$, such that any algorithm that solves \kMVC{} on $G_{k,n}$ in the CONGEST model requires $\Omega(k^2 / \log k \log n)$ rounds, where $n = \size{V}$ can be arbitrarily larger~than~$k$.  
		}
		\newcommand{\ThmKuhnParam}{
	There exists a family of graphs $G_{k,n}(V,E)$, for sufficiently large $k$, such that any algorithm that computes a constant approximation for $\kMVC{}, \kMaxM{},$ or $\kMaxIS$ for $G_{k,n}$ in the LOCAL model requires $\Omega\parentheses{ \sqrt{\log k / \log \log k}}$ rounds, where $n = \size{V}$ can be arbitrarily larger than $k$.
	}
\newcommand{\ThmLocalUB}{
		There exist an $O(k)$ rounds LOCAL algorithms for \kMaxM, \kMaxIS, and any minimization problem $\k{P}$ for $P\in \DLB$.
		}
\begin{document}
	\begin{titlepage}
		\title{Parameterized Distributed Algorithms}
		\author{Ran Ben-Basat\thanks{Technion, Department of Computer Science, \texttt{sran@cs.technion.ac.il}}
			\and Ken-ichi Kawarabayashi \thanks{NII, Japan, \texttt{k\_keniti@nii.ac.jp}, \texttt{greg@nii.ac.jp}} \and Gregory Schwartzman\samethanks}
		\maketitle
		
		\begin{abstract}
%
%
In this work, we initiate a thorough study of parameterized graph optimization problems in the distributed setting. 
In a parameterized problem, an algorithm decides whether a solution of size bounded by a \emph{parameter} $k$ exists and if so, it finds one. 
We study fundamental problems, including Minimum Vertex Cover (\MVC{}), Maximum Independent Set (\MaxIS{}), Maximum Matching (\MaxM{}), and many others, in both the LOCAL and CONGEST distributed computation models.
We present lower bounds for the round complexity of solving parameterized problems in both models, together with optimal and near-optimal upper bounds.

Our results extend beyond the scope of parameterized problems. We show that any LOCAL $(1+\eps)$-approximation algorithm for the above problems must take $\Omega(\eps^{-1})$ rounds. Joined with the algorithm of~\cite{Ghaffari:2017:CLD:3055399.3055471} and the $\Omega\parentheses{\sqrt{\frac{\log n}{\log\log n}}}$ lower bound of~\cite{KuhnMW16}, this settles the complexity of $(1+\eps)$-approximating \MVC{},\MaxM{} and \MaxIS{} at $(\eps^{-1}\log n)^{\Theta(1)}$. 
We also show that our parameterized approach reduces the runtime of exact and approximate CONGEST algorithms for \MVC{} and \MaxM{} if the optimal solution is small, without knowing its size beforehand.
Finally, we propose the first deterministic $o(n^2)$ rounds CONGEST algorithms that approximate \MVC{} and \MaxM{} within a factor strictly smaller than $2$.
		\end{abstract}
		\thispagestyle{empty}
	\end{titlepage}

\section{Introduction}
We initiate the study of \emph{parameterized} distributed algorithms for graph optimization problems, which are fundamental both in the sequential and distributed settings. %
Broadly speaking, in these problems we aim to find some underlying graph structure (e.g., a set of nodes) that abides a set of constraints (e.g., cover all edges) while minimizing the cost.

While parameterized algorithms have received much attention in the sequential setting, not even the most fundamental problems (e.g., Vertex Cover) have a distributed counterpart. We present parameterized upper and lower bounds for fundamental problems such as Minimum Vertex Cover, Maximum Matching and many more. 

\subsection{Motivation}
\ifdefined\fullVersion
In the sequential setting, some combinatorial optimization problems are known to have polynomial time algorithms (e.g., Maximum Matching), while others are NP-complete (e.g., Minimum Vertex Cover). To deal with the hardness of finding an optimal solution to these problems,
\else
To deal with the hardness NP-complete problems,
\fi
the field of parameterized complexity asks what is the best running time we may achieve with respect to some parameter of the problem, instead of the size of the input instance. 
This parameter, usually denoted by $k$, is typically taken to be the size of the solution. Thus, even a running time that is exponential in $k$ may be \mbox{acceptable for small 
\ifdefined\fullVersion
values of
\fi $k$.}

In the distributed setting, a network of nodes, which is represented by a communication graph $G(V,E)$, aims to solve some graph problem with respect to $G$. Computation proceeds in synchronous rounds, in each of which every vertex can send a message to each of its neighbors. The running time of the algorithm is measured as the number of communication rounds it takes to finish. There are two primary communication models: LOCAL, which allows sending messages of unbounded size, and CONGEST, which limits messages to $O(\log n)$ bits (where $n = \size{V}$).

The above definition implies that the notion of "hardness" in the distributed setting is different. Because we do not take the computation time of the nodes into account, we can solve any problem in $O(D)$ rounds of the LOCAL model (where $D$ is the diameter of the graph), and $O(n^2)$ rounds of the CONGEST model. This is achieved by having every node in the graph learn the entire graph topology and then solve the problem.
Indeed, there exist many "hard" problems in the distributed setting, where the dependence on $D$ and $n$ is rather large.

There are several lower bounds for distributed combinatorial optimization problems for both the LOCAL and the CONGEST models. For example, \cite{SarmaHKKNPPW11} provides lower bounds of $\widetilde{\Omega}(\sqrt{n} + D)$ (Where $\widetilde{\Omega}$ hides polylogarithmic factors in $n$), for a range of problems including MST, min-cut, min s-t cut and many more. There are also $\widetilde{\Omega}({n})$ bounds in the CONGEST model for problems such as approximating the network's diameter \cite{abboud2016near} and finding weighted all-pairs shortest paths \cite{Censor-HillelKP17}. Recently, the first near-quadratic (in $n$) lower bounds were shown by \cite{Censor-HillelKP17} for computing exact Vertex Cover and Maximum Independent Set. 
\ifdefined\fullVersion

The above shows that, 
\else
Thus,
\fi 
similar to the sequential setting, the distributed setting also has many "hard" problems that can benefit 
\mbox{from the parameterized complexity lens.}
\ifdefined\fullVersion
Recently, the study of parameterized algorithms  for $\MVC$ and $\MaxM$ was also initiated in the streaming environment by \cite{ChitnisCHM15, ChitnisCEHMMV16}. 
his provides further motivation for our work, showing that indeed non-standard models of computation can benefit from parameterized algorithms.
\fi




\subsection{Our results} 
Given the above motivation we consider the following fundamental problems (See section \ref{sec: prelims} for formal definitions): Minimum Vertex Cover (\MVC), Maximum Independent Set (\MaxIS), Minimum Dominating Set (\MDS), Minimum Feedback Vertex Set (\MFVS), Maximum Matching (\MaxM), Minimum Edge Dominating Set (\MEDS), Minimum Feedback Edge Set (\MFES). 
We use $\mathcal P$ to denote this problem set.

The problems are considered in both the LOCAL and CONGEST model, where we present lower bounds for the round complexity of solving parameterized problems in both models, together with optimal and near-optimal upper bounds. We also extend existing results \cite{KuhnMW16, Censor-HillelKP17} to the parameterized setting. Some of these extensions are rather direct, but are presented to provide a complete picture of the \mbox{parameterized distributed complexity landscape.}

Perhaps the most surprising fact about our contribution, is that our results extend beyond the scope of parameterized problems. We show that any LOCAL $(1+\eps)$-approximation algorithm for the above problems must take $\Omega(\eps^{-1})$ rounds. Joined with the algorithm of~\cite{Ghaffari:2017:CLD:3055399.3055471} and the $\Omega\parentheses{\sqrt{\frac{\log n}{\log\log n}}}$ lower bound of~\cite{KuhnMW16}, this settles the complexity of $(1+\eps)$-approximating \MVC{},\MaxM{} and \MaxIS{} at $(\eps^{-1}\log n)^{\Theta(1)}$. 
We also show that our parameterized approach reduces the runtime of exact and approximate CONGEST algorithms for \MVC{} and \MaxM{} if the optimal solution is small, without knowing its size beforehand.
Finally, we propose the first deterministic $o(n^2)$ rounds CONGEST algorithms that approximate \MVC{} and \MaxM{} within a factor strictly smaller than $2$. Further, for \MVC{}, no such randomized algorithm is known either.

We note that considering parametrized algorithms in the distributed setting presents interesting challenges and unique opportunities compared to the classical sequential environment. In essence, we consider the \emph{communication cost} of solving a parametrized problem on a network. On one hand, we have much more resources (and unlimited computational power), but on the other we need to deal with synchronizations and bandwidth restrictions.

\paragraph{Parametrized problems} We consider combinatorial minimization (maximization) problems where the size of the solution is upper bounded (lower bounded) by $k$. A parameterized distributed algorithm is given a parameter $k$ (which is known to all nodes), and must output a solution of size at most $k$ (at least $k$ for maximization problems) if such a solution exists. Otherwise, \emph{all vertices} must output that no such solution exists. A similar definition is given for parametrized approximation problems (see section~\ref{sec: prelims} for more details).




\subsubsection{Lower bounds} 
We show that the problem of $\MVC$ can be reduced to $\MFVS$ and $\MFES$ via standard reductions which also hold in the CONGEST model. 
There are no known results for $\MFVS$ and $\MFES$ in the distributed setting, so the above reductions, albeit simple, immediately imply that all existing lower bounds for $\MVC$ also apply for $\MFVS$ and $\MFES$. 
Using the fact that $\MFVS$ and $\MFES$ have a global nature, we can achieve stronger lower bounds for the problems. Specifically, we show that no reasonable approximation can be achieved for $\MFVS$ and $\MFES$ 
\ifdefined\fullVersion
in $O(D)$ rounds in the LOCAL model. This is formalized in the following theorem.
\else
\mbox{in $O(D)$ rounds in the LOCAL model:}
\fi
\ifdefined\shortVersion
\vspace{-4mm}
\fi
\begin{theorem-repeat}{thm: MFVS diam lb}
\ThmMFVSLB
\end{theorem-repeat}
\ifdefined\shortVersion
\vspace{-2mm}
\fi


Our main result is a novel lower bound for $(1+\eps)$-approximation for all problems in $\mathcal P$.
Our lower bound states that any $(1+\epsilon)$-approximation (deterministic or randomized) algorithm in the LOCAL model for any problem in $\mathcal P$ requires $\Omega(\epsilon^{-1})$ rounds. Usually, lower bounds in the distributed setting are given as a function of the input (size, max degree), and not as a factor of approximation ratio. Our lower bound also applies to \MaxCut\footnote{In the \MaxCut{} problem, we wish to divide the vertices into two sets $A,B$ such that the number of edges whose endpoints are in different sets is maximized.} and \MaxDiCut\footnote{In the \MaxDiCut{} problem the graph is directed and we wish to divide the vertices into two sets $A,B$ such that the number of directed edges from $A$ to $B$ is maximized.}, whose parametrized variants are not considered in this paper, and thus are not in $\mathcal{P}$.
We state the following theorem.
\ifdefined\shortVersion
\vspace{-2mm}
\fi
\begin{theorem-repeat}{thm:LB-main}
\ThmLBMain
\end{theorem-repeat}
\ifdefined\shortVersion
\vspace{-2mm}
\fi

Our lower bound also has implications for non-parameterized algorithms. 
\ifdefined\fullVersion
The problem of finding
\else
Finding
\fi
a maximum matching   in the distributed setting is a fundamental problem in the field that received much attention \cite{LotkerPR09, LotkerPP15, Fischer17, Bar-YehudaCGS17}. 
\ifdefined\fullVersion
Despite the existence of a variety of approximation algorithms for the problem, 
\else
Yet,
\fi
no non-trivial result is known for computing an exact solution. 

Our lower bounds also has implications for computing a $(1+\epsilon)$-approximation of $\MVC$, $\MaxM$ and \MaxIS{} in the LOCAL model. Combined with the $\Omega\parentheses{ \sqrt{\frac{\log n}{\log\log n}}}$ lower bound of~\cite{KuhnMW16}, we can express a lower bound to the problem as $\Omega\parentheses{\epsilon^{-1} + \sqrt{\frac{\log n}{\log\log n}}} = (\eps^{-1}\log n)^{\Omega(1)}$. Together with the result of~\cite{Ghaffari:2017:CLD:3055399.3055471}, which presents an $(\eps^{-1}\log n)^{O(1)}$ upper bound for the problem, this implies that the complexity of computing a $(1+\eps)$-approximation is given by $(\eps^{-1}\log n)^{\Theta(1)}$. 

Finally, we show a simple and generic way of extending lower bounds to the parameterized setting. 
The problem with many of the existing lower bounds (e.g., \cite{KuhnMW16, Censor-HillelKP17}) is that the size of the solution is $\Tilde{O}(n)$ (linear up to polylogarithmic factors). Thus, it might be the case that if the solution is substantially smaller than the input we might achieve a much faster running time. We show that by simply attaching a large graph to the lower bound graph we can achieve the same lower bounds as a function of $k$, rather than $n$. This allows us to restate our $(1+\eps)$-approximation lower bound and the bounds of \cite{KuhnMW16} and \cite{Censor-HillelKP17} as a function of $k \ll n$. We also show that these lower bounds hold for parameterized problems as defined in this paper. 

\ifdefined\fullVersion
\begin{theorem-repeat}{thm:LB-main-param}
\ThmLBMainParam
\end{theorem-repeat}

\begin{theorem-repeat}{thm:ckp small k}
\ThmCKPParam
\end{theorem-repeat}
    \begin{theorem-repeat}{thm:kuhn small k}
		\ThmKuhnParam  
	\end{theorem-repeat}
\fi
\ifdefined\shortVersion
\vspace{-1mm}
\fi
\subsubsection{Upper bounds} 
We first define the family of problems (\DLB, see section~\ref{sec: prelims} for a formal definition) whose optimal solution is lower bounded by the graph diameter.
If the optimal solution size (\OPT) is small, then for minimization DLB problems we can learn the entire graph in $O(\OPT)$ LOCAL rounds. The problem is actually for the case when the optimal solution is large, and all vertices in the graph must output that no $k$-sized solution exists. Here we introduce an auxiliary result which we use as a building block for all of our algorithms.

\ifdefined\shortVersion
\vspace{-3mm}
\fi
\begin{theorem-repeat}{thm:diameter}
\ThmDiam
\end{theorem-repeat}
\ifdefined\shortVersion
\vspace{-2mm}
\fi

Using the above, we can check the diameter, have all vertices reject if it is too large, and otherwise have a leader learn the entire graph in $O(k)$ rounds.
As for maximization problems (such as \MaxM{} and \MaxIS) the challenge is somewhat different, as the parameter $k$ does not bound the diameter for legal instances.
We first check whether the diameter is at most $2k$ or at least $4k$. If it is bounded by $2k$, we can learn the entire graph. Otherwise, we note that any \emph{maximal} solution has size at least $k$ and is a legal solution to the parameterized problem. Thus, we can efficiently compute a maximal solution by having every node/edge which is a local minimum (according to id) enter the matching/independent set. We repeat this $k$ times and finish 
\ifdefined\fullVersion
(this also works in the CONGEST model). 
We formalize this in the following theorem.
\else
\mbox{(this also works in the CONGEST model).}
\fi

\ifdefined\shortVersion
\vspace{-3mm}
\fi
\begin{theorem-repeat}{thm: local UB}
\ThmLocalUB
\end{theorem-repeat}
\ifdefined\shortVersion
\vspace{-2mm}
\fi

Next, we consider the problems of $\kMVC$ and $\kMaxM$ as case studies for the CONGEST model. We show deterministic upper bound of $O(k^2)$ for $\kMVC$ and $\kMaxM$ (For $\kMVC$ this is near tight according to Theorem~\ref{thm:ckp small k}). We also note that as the complement of an $\kMVC$ is a $\kMaxIS$ we have a near tight upper bound of $O(n-k)^2$ for the problem. 
\ifdefined\fullVersion
This means that if $k$ is large (e.g., $k=n-\log n$) the problem is easy to solve. 
\fi
In the CONGEST model, we first verify that the diameter is indeed small. If it is large, we proceed as we did in the LOCAL model for both problems.
For $\kMVC$, we use a standard kernelization procedure to reduce the size of the graph. This is done by adding every node of degree larger than $k$ into the cover. The remaining graph has a bounded diameter and a small number of edges; thus we use a leader node to collect the entire graph. 
The problem of $\kMaxM$ is more challenging, as we do not use existing kernelization techniques. Instead, we introduce a new augmentation-based approach for the parameterized problem.


We then show how with the help of randomization we can achieve a running time of $O(k+k^2 \log k / \log n)$ for both problems. 
Note that for $k \ll n$ this can be a substantial, up to quadratic, improvement. Further, for \kMVC{} it brings our the round complexity to within a $O(\log^2 k)$ {factor from the lower bound.}

\paragraph{Approximations} We also consider approximation algorithms, in the CONGEST model, for parameterized \MVC{} and \MaxM{}.
We make non-trivial use of the Fidelity Preserving Transformation framework~\cite{FELLOWS201830} and simultaneously apply multiple reduction rules that reduce the parameter from $k$ to $O(k\eps)$. Using this technique, we derive $(2-\eps)$-approximations that run faster than our exact algorithms for any $\eps=o(1)$.
We summarize our other results in Table~\ref{tbl:summary-ub}.

\newcommand{\specialcell}[2][c]{%
\begin{tabular}[#1]{@{}l@{}}#2\end{tabular}}

\newcommand{\rand}{\hfill \ensuremath{[rand.]}}
\renewcommand{\det}{\hfill \ensuremath{[det.]}}

\begin{table*} 
	\centering{
	    \resizebox{0.94 \textwidth}{!}{
				\begin{tabular}{|l|c|c|c|c|c|}
				\hline
				\multirow{2}{*}{\textbf{Variant}} 
				& \multicolumn{2}{|c|}{\textbf{Upper Bound}}
                & \multicolumn{2}{|c|}{\textbf{Lower Bound}}
                \\\hhline{~-----}
                &\multicolumn{1}{|c|}{LOCAL} & \multicolumn{1}{|c|}{CONGEST}&\multicolumn{1}{|c|}{LOCAL} & \multicolumn{1}{|c|}{CONGEST}
                \\
				\hline\hline \multirow{3}{*}{\specialcell[l]{\noindent Exact}} & 
				\multirow{3.5}{*}{$O(k)$ \det}&\multirow{2}{*}{$O\parentheses{k+\frac{k^2\log k}{\log n}}$\qquad{}\rand}
				 & \multirow{3.5}{*}{$\Omega(k)$\quad{}
				 } & \multirow{2}{*}{$\Omega\parentheses{k+\frac{k^2}{\log k\log n}}$ 
				 } \\&&&&\\\hhline{~~~~~~} 
				&  &$O\parentheses{k^2}$\det&& {\scriptsize *\kMVC{} only}\\\hhline{------}
                \multirow{1}{*}{\specialcell[l]{\noindent $(1+\epsilon)$-approx.}} & 
                \multicolumn{2}{|c|}{}&\multicolumn{1}{|l|}{\multirow{2}{*}{\specialcell[l]{\noindent $\Omega\parentheses{\epsilon^{-1}+\sqrt{\frac{\log k}{\log\log k}}}
                $
                }}}&
                \\\hhline{~~~~~~}
				{\small\ \ $\forall \eps=\Omega(1/k)$}
				&\multicolumn{2}{|c|}{}&\multicolumn{1}{|c|}{}&
				\\\hhline{------}
                \multirow{1}{*}{\specialcell[l]{\noindent $(2-\epsilon)$-approx.}} 
                && \multicolumn{1}{|c|}{$O\parentheses{k +\frac{(k\epsilon)^2\log(k\eps)}{\log n}}$\rand} &\multicolumn{2}{|c|}{}\\\hhline{~~~~~~}{\small\ \ $\forall \eps\in[1/k,1]$}
				&&\multicolumn{1}{|c|}{$O(k+(k\epsilon)^2)$\det}&\multicolumn{2}{|c|}{}
				\\\hhline{------}	
				\hhline{------}                  
			\end{tabular}
	}
}
	\normalsize
	\caption{A summary of our round complexity results for \kMVC{} and \kMaxM{}. All lower bounds hold for randomized algorithms as well as deterministic.
	}
	\label{tbl:summary-ub}
		\vspace{-.1cm}
\end{table*} 

\paragraph{Applications to non-parameterized algorithms} We show that our algorithms can also imply faster non-parameterized algorithms if the optimal solution is small, without needing to know its size.
Specifically, we combine our exact and approximation algorithms for parameterized \kMVC{} and \kMaxM{} with doubling and a partial binary search for the value of $k$.
Additionally, our solutions can determine whether to run the existing non-parameterized algorithm or follow the parameterized approach. This results in an algorithm whose runtime is the minimum between current approaches and the number of rounds required for the binary search. Our results are presented in Table~\ref{tbl:summary-general}.

We also present \emph{deterministic} algorithms for $\MVC, \MaxM$ in the CONGEST model with an approximation ratio strictly better than $2$. 
Namely, our 
algorithms terminate in $O(\OPT{}\log \OPT{}) = O(n\log n)$ rounds and provide an approximation ratio of $2-1/\sqrt{\OPT{}}$. 
Here, $\OPT{}$ is the size of the optimal solution and is not known to the algorithm.
These are the first non-trivial $(2-\eps)$-approximation \mbox{results for these problems.}

\begin{table*} 
	\centering{\hspace*{-0.5cm}
		\resizebox{1.05 \textwidth}{!}{
				\begin{tabular}{|l|c|c|c|c|}
				\hline
				& 
				Exact& $(1+\eps)$-approx. & $2$-approx. & {$(2+\eps)$-approx.}
                \\\hline\hline
                \MVC{}&
                \multirow{2.4}{*}{$O\parentheses{\min\set{\OPT{}^2\log \OPT{},n^2}}$}&\multirow{2.4}{*}{$O\parentheses{\OPT{}^2}$}&$O\parentheses{\min\set{\OPT{}\log \OPT{},\frac{\log n\log\Delta}{\log^2\log \Delta}}}$&$O\parentheses{\min\set{\OPT{},\frac{\log\Delta}{\log\log \Delta}}}$\\\hhline{-~~--}
                \MaxM{}&
                &&$O\parentheses{\min\set{\OPT{}\log \OPT{},\Delta+\log^* n}}$&$O\parentheses{\min\set{\OPT{},\Delta+\log^* n}}$\\\hline
			\end{tabular}
	}
}
	\normalsize
	\vspace{-.2cm}
	\caption{Our CONGEST round complexity for deterministic \MVC{} and \MaxM{}, where $\epsilon$ is a positive constant (the actual dependency in $\eps^{-1}$ is logarithmic).
	Here, $\OPT{}$ is the size of the optimal solution and is \emph{not known} to the algorithm.
	}
	\label{tbl:summary-general}
	\vspace{-.2cm}
\end{table*} 

\subsection{Related work}
\paragraph{Distributed Matching and Covering} Both \MVC{} and \MaxM{} have received significant attention in the distributed setting. We survey on the results relevant to this paper. We start with existing lower bounds. In \cite{Censor-HillelKP17} a family of graphs of increasing size is presented, such that computing an \MVC{} for any graph in the family requires $\Omega(n^2 / \log^2 n)$ rounds in the CONGEST model. In \cite{KuhnMW16} a family of graphs is introduced such that any constant approximation for \MVC{} requires $\Omega\parentheses{\min \set{ \sqrt{\log n / \log \log n}, \log \Delta / \log \log \Delta}}$ rounds in the LOCAL model. Both bounds hold for deterministic and randomized algorithms.

For \MVC{}, no non-trivial exact distributed algorithms are known. As for approximations, an optimal (for constant values of $\epsilon$) $(2+\epsilon)$-approximate deterministic algorithm (for the weighted variant) in the CONGEST running in $O(\epsilon^{-1} \log \Delta / \log \log \Delta)$ rounds is given in \cite{Bar-YehudaCS17}. \cite{ImprovedVC} then improved the dependency on $\epsilon$ to $O(\log \Delta / \log \log \Delta + \log \epsilon^{-1} \log \Delta / \log^2 \log \Delta)$, which also results in a faster 2-approximation algorithm by setting $\epsilon=1/n$. In the LOCAL model, a randomized $(1+\epsilon)$-approximation in $(\eps\log n)^{O(1)}$ rounds is due to \cite{Ghaffari:2017:CLD:3055399.3055471}.

For \MaxM{}, there are no known lower bounds for the exact problem, while for approximations, the best known lower bound is due to \cite{KuhnMW16}. No exact non-trivial solution is known for the problem in both the LOCAL and CONGEST models. As for approximations, much is known. We survey the results for the unweighted case. An optimal randomized $(1+\epsilon)$-approximation in the CONGEST model, running in $O(\log \Delta / \log \log \Delta)$ rounds for constant $\epsilon$ is given in \cite{Bar-YehudaCGS17}. As for deterministic algorithms, the best known results in the LOCAL model are due to \cite{Fischer17}, presenting a maximal matching algorithm running in $O(\log^2 \Delta \log n)$ rounds and a $(2+\epsilon)$-approximate algorithm running in $O(\log^2 \Delta \log 1/\epsilon + \log^* n)$. 
As for a deterministic maximal matching in the CONGEST, to the best of our knowledge, the best known approach is to color the edges and go over each color class, resulting in a running time of $O(\Delta + \log^*n)$ as proposed in \cite{Bar-YehudaCGS17}. 

\paragraph{Distributed parameterized algorithms} Parameterized distributed algorithms were previously considered for detection problems. Namely, in \cite{KorhonenR17} it was shown the detecting $k$-paths and trees on $k$-nodes can be done deterministically in $O(k2^k)$ rounds of the broadcast CONGEST model. Similar, albeit randomized, results were obtained independently by \cite{EvenFFGLMMOORT17} in the context of distributed property testing \cite{Censor-HillelFS16}.  


\vspace{-2mm}
\section{Preliminaries}
\label{sec: prelims}
\vspace{-2mm}
In this paper, we consider the classic and parameterized variants of several popular graph packing and covering problems, in the LOCAL and CONGEST models. 
A solution to these problems is either a vertex-set or an edge-set. In vertex-set solutions, we require that each vertex will know if it is in the solution or not. For edge-set problems, each vertex must know which of its edges are in the solution, and both end-points of an edge must agree.
\ifdefined\fullVersion
Computation takes place in synchronous rounds during which each node first receives messages from its neighbors, then perform local computation, and finally send messages to its neighbors. 
Each of the messages sent to a node neighbor may be different from the others, while the size of messages is unbounded in the LOCAL model or of size $O(\log n)$ in the CONGEST. In both models, the communication graph is identical to the graph on which the problem is solved. That is, two nodes may send messages to each other only if they share an edge.
\fi

\newcommand{\problem}[1]{\textbf{#1.}\quad}
\ifdefined\shortVersion
We use $\mathcal P = \set{\text{\MVC, \MaxM, \MaxIS, 
\MDS, \MEDS, \MFVS, \MFES}}$ to denote the set of problems considered in this work.
For completeness,  we give formal definitions of the problems in Appendix~\ref{app:problems}.
\else
Let $G=(V,E)$ (Directed graph for \MFES) denote the target graph. 
We consider $U\subseteq V$ to be a feasible solution to the following vertex-set problems if:
\begin{itemize}
\item \problem{Minimum Vertex Cover (\MVC)}
$\forall e\in E: e\cap U\neq\emptyset$.
\item \problem{Maximum Independent Set (\MaxIS)}
$\forall u,v\in U: \set{u,v}\notin E$.
\item \problem{Minimum Dominating Set (\MDS)}
$\forall v\in V, \exists u\in U: \set{u,v}\in E$.
\item \problem{Minimum Feedback Vertex Set (\MFVS)}
$G[V\setminus U]$ is acyclic.
\end{itemize}
Next, we call $Y\subseteq E$ a feasible solution to the following edge-set problems if:
\begin{itemize}
\item \problem{Maximum Matching (\MaxM)}
$\forall v,u,w\in V: \set{v,u}\in Y\implies \set{u,w}\notin Y$.
\item \problem{Minimum Edge Dominating Set (\MEDS)} 
\mbox{$\forall e\in E, \exists e'\in Y: e\cap e' \neq \emptyset$.}
\item \problem{Minimum Feedback Edge Set (\MFES)} 
\mbox{$(V,E\setminus Y)$ is acyclic.}
\end{itemize}
We use $\mathcal P = \set{\text{\MVC, \MaxM, \MaxIS, 
\MDS, \MEDS, \MFVS, \MFES}}$ to denote the above set of problems.
\fi
%
%
Given a parameter $k$, a parameterized algorithm computes  a solution of size bounded by $k$ if such exists; otherwise, the nodes must report that no such solution exists. For a problem $P\in \mathcal P$, we denote by \k{$P$} (e.g., \kMVC) the parameterized variant of the problem. We use \k{$\mathcal P$} to denote all these problems.
We note that all nodes know $k$ when the algorithm starts and that the result may not be the optimal solution to the problem.
\begin{definition}
An algorithm for a \emph{minimization} (respectively, \emph{maximization}) problem $\k{P}$ must find a solution of size at most (respectively, at least) $k$ if such exists. If no such solution exists then all nodes must report so when the algorithm terminates.
\end{definition}
We now generalize our definition to account for randomized and approximation algorithm.
\begin{definition}
For $\alpha\ge1$, an $\alpha$-approximation algorithm for a \emph{minimization} (respectively, \emph{maximization}) problem $\k{P}$ must find a solution of size at most $\alpha k$ (respectively, at least $k/\alpha$) if a solution of size $k$ exists. Otherwise, all nodes must report that no $k$-sized solution exists.
\end{definition}
\begin{definition}
Given some $\delta\ge0$ and $\alpha\ge1$, an $\alpha$-approximation Monte Carlo algorithm for a problem $\k{P}$ terminates with an $\alpha$-approximate solution to $\k{P}$ with probability at least $1-\delta$.
\end{definition}

We now define the notion of \emph{Diameter-Lower-Bounded} (\DLB) problems. Intuitively, this class contains all problems whose optimal solution size is $\Omega(D)$, which allows efficient algorithms for their parameterized variants. For example, \DLB{} includes \MVC, \MaxIS, \MaxM, \MDS, and \MEDS{}, but not \MFVS{} and \MFES. Roughly speaking, these problems admit efficient LOCAL parameterized algorithms as the parameter limits the radius that each node needs to see for solving the problem.
\begin{definition}
An optimization problem $P$ is in \DLB{} if for any input graph $G$ of diameter $D$, the size of an optimal solution to $P$ is of size $\Omega(D)$.
\end{definition}
\section{Lower Bounds}
\ifdefined\fullVersion
In this section, we present lower bounds for a large family of classical and parameterized distributed graph problems.
\fi

\subsection{Non-parametrized distributed problems}
Here, we provide a construction that implies lower bounds for approximating all problems in $\mathcal{P}$.
Our lower bounds dictate that any algorithm that computes a $(1+\epsilon)$-approximation, for $\epsilon=\Omega(1/n)$, requires at least $\Omega(\epsilon^{-1})$ rounds in the LOCAL model, even for randomized algorithms. We note that for all these problems, no superlogarithmic lower bound (which we can get, e.g., for $\epsilon=n^{-2/3}$) was known for approximations. Further, for some problems, such as \MaxM, no such lower bound is known even for exact solutions in the CONGEST model. 

We then generalize our approach and show that even in the parameterized variants of the problems (where the optimal solution is bounded by $k$), $\Omega(\epsilon^{-1})$ rounds are needed for an arbitrarily large graphs (where $n\gg k$).
Our approach is based on the observation that for any $x$-round algorithm, there exists an input graph of many distinct $\Theta(x)$-long paths, such that the algorithm has an additive error on each of the paths, which accumulates over all paths in the construction. Intuitively, we show that for any set of $n$ node identifiers and any algorithm that takes $o(\epsilon^{-1})$ rounds, it is possible to assign identifiers to nodes such that the approximation ratio would be larger than $1+\epsilon$. Our goal in this section is \mbox{to prove the following theorem.}

\begin{theorem}\label{thm:LB-main}
\ThmLBMain
\end{theorem}

\subsubsection{Basic Construction}
We start with lower bounds for the non-parameterized variants of the problems, where the optimal solution may be of size $\Theta(n)$.
For integer parameters $r,\ell > 10$, 
the graph 
$$G_{r, \ell}=\parentheses{\set{v_0}\cup\set{v_{i,j}\mid i\in[r], j\in[\ell+1]},\set{\set{v_0,v_{i,0}}\mid i\in[r]}\cup \set{\set{v_{i,j},v_{i,j+1}}\mid i\in[r], j\in[\ell]}}$$ consists of $r$ disjoint paths of length $\ell$, whose initial nodes are connected to a central vertex $v_0$. We also consider the digraph $\Vec G_{r,\ell}$ in which each edge is oriented away from $v_0$; i.e., 
$$\Vec G_{r, \ell}=\parentheses{\set{v_0}\cup\set{v_{i,j}\mid i\in[r], j\in[\ell+1]},\set{\parentheses{v_0,v_{i,0}}\mid i\in[r]}\cup \set{\parentheses{v_{i,j},v_{i,j+1}}\mid i\in[r], j\in[\ell]}}$$
We present our construction in Figure~\ref{fig: basic construction}.
\begin{figure}
\centering
\ifdefined\shortVersion
  \includegraphics[width=.5\linewidth]{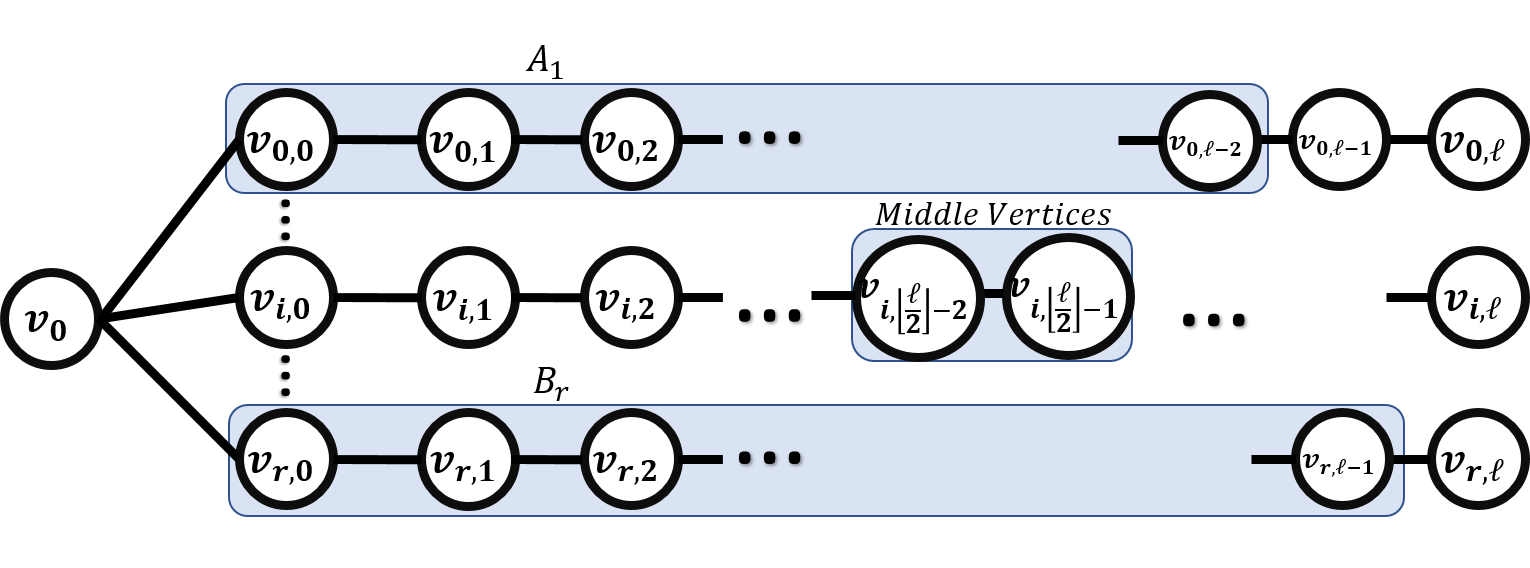}
\else
  \includegraphics[width=.7\linewidth]{basic_constructionV2.png}
\fi
\vspace{-2mm}
  \caption{Our basic construction. The middle vertices' ($v_{i,\floor{\ell/2}-2}$ and $v_{i,\floor{\ell/2}-1}$) output remains the same if we reverse the identifiers along $A_i$ and the algorithm takes less than $\floor{\ell/2}-3$ rounds. The output of $v_{i,\floor{\ell/2}}$ also remains the same if $B_i$ is flipped. However, since the distance of these vertices from $v_0$ change, the output is suboptimal for at least one of the orderings. 
  }
  \label{fig: basic construction}
\end{figure}
$G_{r, \ell}$ has $n=r(\ell+1)+1$ vertices and a diameter (as $r>1$) of $2\ell$.
Observe that the optimal solutions of  \MVC, \MaxM, \MaxIS, \MDS, and \MEDS{} on $G_{r, \ell}$ (for $r,\ell\ge 3$) have values of $\Theta(r\ell)=\Theta(n)$. 
For every path $i\in[r]$, let $A_i=\angles{v_{i,0},\ldots,v_{i,\ell-2}}$ and $B_i=\angles{v_{i,0},\ldots,v_{i,\ell-1}}$ denote the longest odd and even length subpaths that do not include $v_0$ and $v_{i,\ell}$.
Given a path $P$ of vertices with assigned identifiers, we denote by $P^R$ a reversal in the order of identifiers. For example, if the identifiers assigned to $A_0$ were $\angles{0,\ldots,\ell-2}$, then those of $A_0^R$ would be $\angles{\ell-2,\ell-3,\ldots,0}$. 
This reversal of identifiers along a path plays a crucial role in our lower bounds. Intuitively, if the number of rounds is less than $\ell/2-3$ and we reverse $A_i$ or $B_i$, the output of the middle vertices 
\emph{would change} to reflect the mirror image they observe.
We show that this implies that on either the original identifier assignment or its reversal, the algorithm must find a sub-optimal solution to the $i$'th path (where the choice of whether to flip $A_i$ or $B_i$ depends if the output is a vertex set or edge set). In turn, this would sum up to a solution that is far from the optimum by at least an $r$-additive factor. As the optimal solution is of size $\Theta(r\ell)$, this implies a multiplicative error of $1+\Theta(r/(r\ell)) = 1+1/\ell$.

We show that for arbitrarily large graphs with an optimal solution of size $\Theta(n)$, any $x$-round algorithm must have an additive error of $\Omega(n/x)$.
\begin{lemma}\label{lem:detApproxLB}
Let $x,r\in\mathbb N^+$ be integers larger than $10$,
and let $\mathcal I$ be a set of $n=(2x+4)r+1$ node identifiers.
For any deterministic LOCAL algorithm for \MVC, \MaxM, \MaxIS, \MDS, \MEDS{}, \MaxDiCut, or \MaxCut{}
that terminates in $x$ rounds on $G_{r,2x+3}$, there exists an assignment of vertex identifiers for which the algorithm has an additive error of $\Omega(n/x)$.
\end{lemma}
\begin{proof}
First, let us characterize the optimal solutions for each of the problems on $G_{r,2x+3}$ (or $\Vec G_{r,2x+3}$ for \MaxDiCut{}). 
For simplicity of presentation, we assume that $\mathcal I=[n]$ and $(x\mod 6) = 0$ although the result holds for any $\mathcal I$ and $x$.
\ifdefined\shortVersion
For completeness, Appendix~\ref{app:optSolutions} provides the characterization of the optimal solutions.
\else
We have 
\begin{align}
&\OPT{}_{MVC}=\set{v_{i,2j}\mid i\in[r], j\in[x+2]}, 
\\&\OPT{}_{MM}=\set{\set{v_{i,2j},v_{i,2j+1}}\mid i\in[r], j\in[x+2]}, 
\\&\OPT{}_{MaxIS}=\OPT{}_{\MaxCut}=\OPT{}_{\MaxDiCut}=\set{v_0}\cup \set{v_{i,2j+1}\mid i\in[r], j\in[x+2]}, \footnotemark{}
\\&\OPT{}_{MDS}=\set{v_{i,3j}\mid i\in[r], j\in\brackets{{2x/3+1}}}, 
\\&\OPT{}_{MEDS}=\set{\set{v_{0},v_{i,0}}\mid i\in[r]}\cup \set{\set{v_{i,3j+2},v_{i,3j+3}}\mid i\in[r], j\in\brackets{{2x/3}}}.\footnotemark{}
\end{align}
\addtocounter{footnote}{-1}
\footnotetext{For \MaxCut{}, the complement solution $V\setminus\set{v_0}\cup \set{v_{i,2j+1}\mid i\in[r], j\in[x+2]}$ is also optimal, but the correctness would follow from similar arguments.
}
\stepcounter{footnote}
\footnotetext{Each edge of the form $\set{v_0,v_{i,0}}$ (for $i\in[r]$) may each be replaced by the $\set{v_{i,0},{v_{i,1}}}$ edge. Unlike the other problems, the optimal solution here is not unique, but this does not affect the proof. 
}

For example, this means that the only optimal solution to \MVC{} picks all vertices whose distance from $v_0$ is odd.
\fi
\begin{figure}
\centering
  \includegraphics[width=.6\linewidth]{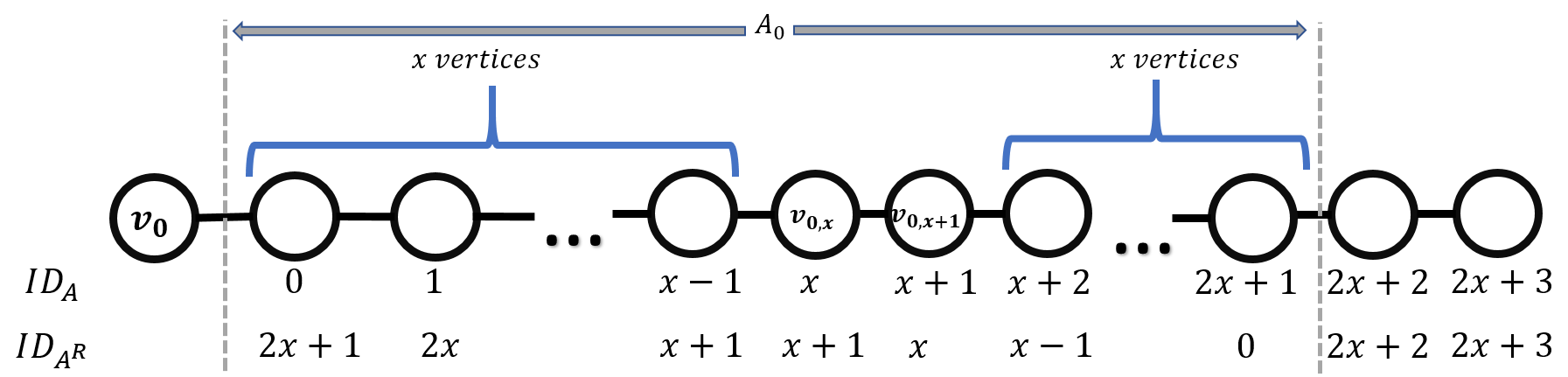}
  \caption{Before the reversal (in $A_0$), vertices $v_{0,x}$ and $v_{0,x+1}$ have identifiers $x$ and $x+1$. When reversing $A_0$, the vertices switch identifiers and must switch their output. That is, $out_{A}(v_{0,x})=out_{A^R}(v_{0,x+1})$ and $out_{A}(v_{0,x+1})=out_{A^R}(v_{0,x})$ in any algorithm that takes fewer than $x$ communication rounds.}
  \label{fig: reversal}
\end{figure}
Next, consider a path $i\in[r]$ and consider the case where every vertex $v_{i,j}$ has identifier $(2x+4)i+j$. 
From the point of view of the node with identifier $(2x+4)i+x$ (which is $v_{i,x}$ in this assignment), in its $x$-hop neighborhood it has nodes with identifiers $(2x+4)i,(2x+4)i+1,\ldots,(2x+4)i+x-1$ on one port (side) 
and identifiers $(2x+4)i+x+1,(2x+4)i+x+2,\ldots,(2x+4)i+2x+2$ on the other. On the other hand, if we reverse $A_i$ (i.e., assign $v_{i,j}\in A_i$ with identifier $(2x+4)i+2x+1-j$), the view of $v_{i,x}$ remains exactly the same. 
That is, the node observes the exact same topology and vertex identifiers in both cases.
Since the algorithm is deterministic, the output of $(2x+4)i+x$ must remain the same for both identifier assignments, even though now it is placed in $v_{i,x+1}$! Similarly, reversing $A_i$ would mean that the node with identifier $(2x+4)i+x+1$ (which changes places from $v_{i,x+1}$ to $v_{i,x}$ after the reversal) also provides the same output in both cases. 
Therefore, reversing $A_i$ would switch the outputs of $v_{i,x+1}$ and $v_{i,x}$. This implies that the output of the algorithm is suboptimal for either $A_i$ or $A_i^R$ for \MVC, \MaxM{}, \MaxIS, \MaxCut{}, and \MaxDiCut.
We illustrate this reversal on path $0$ in Figure~\ref{fig: reversal}.
Repeating this argument for $B_i$, we get that its reversal would switch the outputs of $v_{i,x}$ and $v_{i,x+2}$, making the algorithm err in \MDS{} \mbox{(as $v_{i,x}$ is in the optimal cover while $v_{i,x+2}$ is not).}

For MEDS, every $u,v$ that share an edge must agree whether it is in the solution or not. 
In an optimal solution the edge $\set{v_{i,x},v_{i,x+1}}$ must be in the dominating set while $\set{v_{i,x+1},v_{i,x+2}}$ must not. However, by reversing $B_i$ the identifiers of $v_{i,x}$ and $v_{i,x+2}$ switch, changing edge added from $\set{v_{i,x}, v_{i,x+1}}$ to $\set{v_{i,x+1}, v_{i,x+2}}$ or vice versa, implying an error for $B_i$ or $B^R_i$.

\ifdefined\fullVersion

\fi
As we showed that there exists an identifier assignment that ``fools'' the algorithm on every path $i\in[r]$, we conclude that the algorithm has an additive error of at least \mbox{$r= (n-1)/(2x+4) = \Omega(n/x)$.\qedhere}
\end{proof}
Lower bounds for MFVS and MFES follow from the reduction to $\MVC$ in Theorem~\ref{thm: mvc reduction}.
\ifdefined\fullVersion

\fi
Since the optimal solution to all problems on $G_{r,2x+4}$ is of size $\Theta(n)$, we have that the algorithms have an approximation ratio of $1+\Theta((n/x)/n) = 1+\Theta(1/x)$. Plugging $\epsilon = \Theta(1/x)$ \mbox{we conclude the following.}

\begin{corollary}
For any $\epsilon=\Omega(1/n)$, any deterministic LOCAL algorithm that computes a $(1+\epsilon)$-multiplicative approximation for any $P\in\mathcal P$ requires $\Omega(\epsilon^{-1})$ rounds.
\end{corollary}

\ifdefined\shortVersion
In Appendix~\ref{app:minimax}, we extend this argument by applying Yao's Minimax Principle~\cite{minimax} and conclude the correctness of Theorem~\ref{thm:LB-main}.
\else
Next, we use Yao's Minimax Principle~\cite{minimax} to extend the lower bound to randomized algorithms and prove Theorem~\ref{thm:LB-main}.
\renewcommand*{\proofname}{Proof of Theorem~\ref{thm:LB-main}}
\begin{proof}
In the proof of Lemma~\ref{lem:detApproxLB}, we showed that there is an identifier assignment that forces any deterministic algorithm to solve each path sub-optimally. For randomized algorithms, this does not necessarily work. However, we show that by applying Yao's Minimax Principle~\cite{minimax} we can get similar bounds for Monte Carlo algorithms. That is, we show a probability distribution over inputs such that every deterministic algorithm which executes for $O(\epsilon^{-1})$ rounds incurs a $(1+\epsilon)$-multiplicative error with high probability. 

Intuitively, we consider an input distribution that randomly selects for each path whether to use the original ordering or whether to reverse $A_i$ or $B_i$ (this depends on the problem at hand, as in Lemma~\ref{lem:detApproxLB}). In essence, this gives the algorithm a chance of at most half of finding the optimal solution to each path.

Formally, fix an initial vertex identifier assignment. 
For the problems of \MVC, \MaxM{}  and \MaxIS{} (respectively, \MDS{} and \MEDS) consider the $2^r$ inputs obtained by reversing any subset of $\set{A_i}_{i\in [r]}$ (respectively, $\set{B_i}_{i\in [r]}$). Next, consider the uniform distribution that gives each of these inputs probability $2^{-r}$. 
A similar argument to that of Lemma~\ref{lem:detApproxLB} shows that the probability an algorithm computes the optimal solution to each path $i\in[r]$ is at most half.
We now apply a simplified version of the Chernoff Bound that states that for any $\gamma\in(0,1)$, $r\in \mathbb N^+$, a binomial random variable $X\sim \text{Bin}(r,1/2)$ satisfies $\Pr\brackets{X\ge r/2(1+\gamma)}\le 2EXP(-r\gamma^2/6)$. Plugging $\gamma=\sqrt{-6\logp{(1-2\delta)/2}/r}$ we get
\begin{align*}
\Pr\brackets{X\ge r/2 + \sqrt{-6r\logp{(1-2\delta)/2}/4}}\le 2EXP(\logp{(1-2\delta)/2}
)=1-2\delta.
\end{align*}
That is, any deterministic algorithm does errs on at least $r/2 - \sqrt{-6r\logp{(1-2\delta)/2}/4}$ of the paths with probability at least $1-2\delta$ (for inputs chosen according to the above distribution).
Next, we restrict the value of the error probability to $\delta \le 1/2-EXP(-2r/9)=1/2-EXP(-\Omega(n\epsilon))$, which guarantees that 
$$r/2 + \sqrt{-6r\logp{(1-2\delta)/2}/4} \le r/2 + \sqrt{-6r(-2r/9)/4}=5r/6.
$$
This means that on at least $r/6$ of the paths the algorithm fails to find the optimal solution and adds an additive error of at least one.
Therefore, since the optimal solution is of size $\Theta(n)$, the approximation ratio obtained by the deterministic algorithm is 
\begin{align*}
\alpha=1 + \Omega\parentheses{\frac{r}{n}}
= 1 + \Omega\parentheses{\frac{r}{r\ell}}  = 1+\Omega(1/x).
\end{align*}
Thus, by the Minimax principle we have that any Monte Carlo algorithm with less than $x$ rounds also have an approximation $1+\Omega(1/x)$. Using $x=\Theta(\epsilon^{-1})$, we established the \mbox{correctness of theorem. 
\qedhere}
\end{proof}
\renewcommand*{\proofname}{Proof}
\fi
\subsection{Lower bounds for $\MFVS$ and $\MFES$}
	We first consider the problems of finding a $\MFVS$ and $\MFES$ in a graph. 
	We first show an $\Omega(D)$ lower bound for both problems (even for approximation), and then show that via standard reductions any distributed lower bound for $\MVC$ is also a lower bound for $\MFVS$ and 
	\ifdefined\fullVersion
	$\MFES$ (even for approximations).
	\else
	\MFES{}.
	\fi
	
	
	\ifdefined\fullVersion
	\paragraph{An $\Omega(D)$ lower bound in the LOCAL model, even for approximations}
	\fi
	For parameters $k,t,d\in\mathbb N^+$ such that $d>10$, Consider the graph $G_1=(V_1,E_1)$ and (respectively, the digraph $\Vec G_2=(V_2,\Vec E_2)$) that consist of a star with $k$ outer nodes, in which each outer node (edge) is connected to $t$ cycles (directed cycles) of length $d$; formally:
	\ifdefined\fullVersion
	{ \begin{align*}
	V_1 &= V_2 = 
	\set{v'}\cup \set{v_i\mid i\in[k]}\cup\set{v_{i,j,\ell}\mid i\in[k],j\in[t],\ell\in[d]},\\
	E_1 &=\set{\set{v',v_{i}}\mid i\in[k]} \cup 
	\set{\set{v_i,v_{i,j,0}},\set{v_i,v_{i,j,d-1}}\mid i\in[k], j\in[t]} \\&\qquad{}\cup
	\set{\set{v_{i,j,\ell},v_{i,j,\ell+1}}\mid i\in[k],j\in[t],\ell\in[d-1]},\\
	\Vec E_2 &=\set{(v',v_{i})\mid i\in[k]} \cup 
	\set{(v_{i},v_{i,j,0}),(v_{i,j,d-1},v')\mid i\in[k], j\in[t]} \\&\qquad{}\cup
	\set{(v_{i,j,\ell},v_{i,j,\ell+1})\mid i\in[k],j\in[t],\ell\in[d-1]}.
	\end{align*}}
	\else
	{\small \begin{align*}
	V_1 &= V_2 = 
	\set{v'}\cup \set{v_i\mid i\in[k]}\cup\set{v_{i,j,\ell}\mid i\in[k],j\in[t],\ell\in[d]},\\
	E_1 &=\set{\set{v',v_{i}}\mid i\in[k]} \cup 
	\set{\set{v_i,v_{i,j,0}},\set{v_i,v_{i,j,d-1}}\mid i\in[k], j\in[t]} \cup
	\set{\set{v_{i,j,\ell},v_{i,j,\ell+1}}\mid i\in[k],j\in[t],\ell\in[d-1]},\\
	\Vec E_2 &=\set{(v',v_{i})\mid i\in[k]} \cup 
	\set{(v_{i},v_{i,j,0}),(v_{i,j,d-1},v')\mid i\in[k], j\in[t]} \cup
	\set{(v_{i,j,\ell},v_{i,j,\ell+1})\mid i\in[k],j\in[t],\ell\in[d-1]}.
	\end{align*}}
	\fi
	Observe that the sizes of the graphs are $|V_1| = |V_2| = k (1+ d t)+1$, and that they have a unique optimal solutions of size $k$ -- the \MFVS{} of $G_1$ is $\set{v_i\mid i\in[k]}$ and the \MFES{} of $\Vec G_2$ is $\set{(v',v_{i})\mid i\in[k]}$.
	\ifdefined\fullVersion
	
	Next, consider the trees we obtain from deleting the middle edge from each cycle; namely: 
	\begin{align*}
	T_1&=\parentheses{V_1,E_1\setminus\set{\set{v_{i,j,\floor{d/2}},v_{i,j,\floor{d/2+1}}}\mid i\in[k],j\in[t]}},\\ \Vec T_2&=\parentheses{V_2,\Vec E_2\setminus\set{\parentheses{v_{i,j,\floor{d/2}},v_{i,j,\floor{d/2+1}}}\mid i\in[k],j\in[t]}}.\qquad{}\qquad{}
	\end{align*}
	\else
	Next, consider the trees $T_1,\Vec T_2$ we obtain from deleting the middle edge from each cycle.
	\fi
	Clearly, the minimal feedback set for trees is the empty set. Now, if the algorithm takes at most $d/2-2=\Omega(D)$ rounds, the vertices $\set{v_i\mid i\in[k]}$ cannot distinguish between $G_1$ and $T_1$ and between $G_2$ and $T_2$. 
	Therefore, the output of $\set{v_i\mid i\in[k]}$ must be identical in both cases, which means that any algorithm must have an additive error of at least $kt/2=\Omega(|V|/D)$ at least for one of the inputs. We summarize this in the following theorem.
	\begin{theorem}
	\label{thm: MFVS diam lb}
	\ThmMFVSLB
	\end{theorem}

\ifdefined\shortVersion
In Appendix~\ref{app:MVC2MFVS}, we prove the following reduction from \MVC{} to \MFES{} and \MFVS{}.
	\begin{theorem}
	\label{thm: mvc reduction}
	Any distributed $(\alpha, \beta)$-approximation algorithm in the CONGEST model for $\MFVS$ or $\MFES$ is also an $(\alpha, \beta)$-approximation algorithm in the CONGEST model for $\MVC$ running in the same number of rounds.
	\end{theorem}
\else  
	\subsubsection{Reducing $\MVC$ to $\MFVS, \MFES$}
	We present reductions from $\MVC$ to $\MFVS$ and $\MFES$. Despite being folklore, we present them here for completeness and to point out that they work in the distributed setting without issues (no congestion is incurred).
	
	Let $G(V,E)$ be an instance of \MVC{} and define by $G_1$ the graph in which we replace each edge with a triangle; formally, $G_1=(V\cup E, \set{\set{v,u},\set{v,e},\set{u,e}\mid e=\set{u,v}\in E})$. 
	Further, let $\Vec G_2=(\set{v_{in},v_{out}\mid v\in V}, \set{(v_{in},v_{out})\mid v\in V}
	\cup\set{(u_{out},v_{in}),(v_{out},u_{in})\mid \set{u,v}\in E})$ be the digraph in which each $v\in V$ is replaced with two vertices $v_{in},v_{out}$ connected by an arc and each edge $\set{u,v}$ is replaced by the arcs $(u_{out},v_{in}),(v_{out},u_{in})$.
	We prove the following lemma:
	\begin{restatable}{lemma}{testname}
		There is an $\MVC$ of size $k$ in $G$ if and only there is an $\MFVS$ of size $k$ in $G_1$. Further, There is an $\MVC$ of size $k$ in $G$ if and only there is an $\MFES$ of size $k$ in $G_2$
	\end{restatable}
	\begin{proof}
	Assume that there is a cover $U\subseteq V$ of $G$, then $U$ is an FVS for $G_1$ and $\set{(v_{in},v_{out})\mid v\in U}$ is an FES for $G_2$. Similarly, let $F_1$ be an FVS for $G_1$ and let $F_2$ be a $FES$ for $G_2$. Denote by $F_1'$ the vertex set we get by arbitrarily replacing each $\set{u,v}\in F_1$ by one of its endpoint (e.g., $u$); then $F_1'$ is a cover of $G$. 
	Similarly, the set $F_2'=\set{v\mid ((v_{in},v_{out})\in F_2) \vee (\exists u:(v_{out},u_{in})\in F_2)}$ is a cover for $G$ (of size at most $|F_2|$). This concludes the proof.
	\end{proof}

	\paragraph{Distributed implementation} We note that the above reductions can be simulated in the distributed setting without incurring congestion. In \MFVS{}, $\set{u,v}$ can be simulated by either $u$ or $v$, and in the second reduction a node $v\in V$ can simulate $v_{in}, v_{out}$. 

	The above lemma implies a reduction from \MVC{} to \MFVS{} and \MEDS{} even for $(\alpha,\beta)$-approximations. We say that a solution with value $X$ to some optimization problem is an $(\alpha, \beta)$-approximate solution if it holds that $X \in [\alpha^{-1} \OPT{} - \beta, \alpha \OPT{} + \beta]$. Finally, we state \mbox{our main theorem}:
	
	\begin{theorem}
	\label{thm: mvc reduction}
	Any distributed $(\alpha, \beta)$-approximation algorithm in the CONGEST model for $\MFVS$ or $\MFES$ is also an $(\alpha, \beta)$-approximation algorithm in the CONGEST model for $\MVC$ running in the same number of rounds.
	\end{theorem}
\fi	
\subsection{Parameterized Problems}
In this section, we provide lower bounds for parameterized problems. 
The above lower bounds apply for a broad set of problems that have been extensively studied in the distributed setting and are likely to be of independent interest.
However, as the optimal solution to all problems is linear in the graph size, it fails to provide us with bounds for parameterized computation. Namely, if there exists an optimal solution of size $k$ (which could be significantly smaller than $n$), how does that affect the lower bounds? In this section, we prove that the above theorem holds even for graphs that are arbitrarily larger than $k$.

We achieve the above by showing a simple way to extend lower bounds where the size of the solution is large to the case where it is much smaller than the size of the input. This allows us to extend the lower bounds of \cite{Censor-HillelKP17} and \cite{KuhnMW16} to the parameterized setting.

    We achieve the above by \emph{attaching} (see definition \ref{def: attach graph}) some very large graph to the lower bound graph. We pick the graph to attach in such a way that the size of the solution stays (almost) the same. The graphs we attach are $K_n$, the clique on $n$ vertices and $S_n$, the star graph on $n$ vertices. Formally, Let $S_n(V_S,E_S)$ be a \emph{star} of size $n$, where $V_S =\set{c} \cup \set{v_1,...,v_n}, E_S={(c,v_i) \mid i \in [n]}$. Let us formally define the concept of graph attachment.
    
    \begin{definition}
    \label{def: attach graph}
    Let $G(V,E), H(V_H,E_H)$ be any two non-empty graphs. We define $G'(V',E')$ to be the graph resulting from \emph{attaching} $H$ to $G$ via $(v,u)$ such that $v \in V, u\in V_H$. Formally, $V' = V \cup V_H, E'=E\cup E_H \cup \set{(v,u)}$.
    \end{definition}
    
    Next, let us denote by $P(G)$ the size of the optimal solution to $G$ for some $P \in \mathcal P$.  
 %
%
    %
     We state the following theorem.
    %
    \begin{theorem}
    \label{thm: generalize LB}
     For any $P \in \mathcal{P}$ the following holds: If there exists a graph $G(V,E), k=\size{V}$ where vertices have unique identifiers and some vertex $u$ has the identifier $0$, and any algorithm computing an $\alpha$-approximate solution to $P$ on $G$ requires $\Omega(f(k))$ rounds, then there exists a graph $G'(V',E'), \size{V'} \gg k$ such that every algorithm computing an $(\alpha, \alpha-1)$-approximate solution to $P$ on $G'$ requires $\Omega(f(k))$ rounds in the local model and $f(k)\cdot \log k / \log n$ \mbox{rounds in the CONGEST model.}

    \end{theorem}
    \begin{proof}
    We prove the claim for $\MVC$ and $\MaxIS$. The bound follows for $\MFVS,\MFES$ via reductions to $\MVC$ (Theorem~\ref{thm: mvc reduction}), while for the rest of the problems the proof is identical up to~notation.
    
    Let $G'(V',E'), G''(V'',E'')$ be the graphs created from $G$ by attaching $S_n, K_n$ to $u \in V$, for $n\geq 2$.
    First, note that since $n \geq 2$ we have $\MVC(G') = \MVC(G) + 1$ and $\MaxIS(G'') = \MaxIS(G) + 1$. This is because we can always extend any valid solution to $G$ to a solution to $G',G''$ by taking a single vertex from the attached graph, and we must do so to achieve a valid \mbox{optimal solution to $G', G''$.}
    
    
    We execute the algorithm for $G', G''$ on $G$ where the vertex $u$ (which has id $0$) simulates the attached star/clique. 
    All vertices in $V$ return the result of the simulation and the output of the $S_n$/$K_n$ vertices is ignored.
    This simulation does not incur any congestion. We note that the messages in $G', G''$ are larger, thus simulating a single message in $G$ requires $O(\log n / \log k)$ rounds. Thus, the algorithm for $G',G''$ may potentially be faster by an $O(\log n / \log k)$ factor.
    
    As for correctness, let $C', C''$ be the sizes of the solutions returned by the algorithm on $G', G''$. 
    It holds that $C' \leq \alpha(\MVC(G)+1)$ and $C'' \geq \alpha^{-1}(\MaxIS(G)+1)$. Because we attached a star/clique there must be some extra vertex in the solution that gets removed when converting it to a solution to $G$. Thus, the returned solutions are bounded by $\alpha \MVC(G) + \alpha - 1$ and $\alpha^{-1} \MaxIS(G) + \alpha^{-1} - 1$.\qedhere
    \end{proof}
    
    We note that there are bounds for computing an exact \MVC{} \cite{Censor-HillelKP17} and any constant approximation \cite{KuhnMW16} in terms of $n$, the size of the graph. Namely, \cite{Censor-HillelKP17}  presents a family of graphs of increasing size, such that computing an \MVC{} for any graph in the family requires $\Omega(n^2 / \log^2 n)$ rounds in the CONGEST model. In \cite{KuhnMW16} a family of graphs is presented such that any constant approximation for \MVC{} requires $\Omega\parentheses{\min \set{ \sqrt{\log n / \log \log n}, \log \Delta / \log \log \Delta}}$ rounds in the LOCAL model. Both bounds hold for deterministic and randomized algorithms.
    In both cases, the size of the optimal solution in the lower bound graph is $\Tilde{\Theta}(n)$. We use Theorem~\ref{thm: generalize LB} to extend these bounds for the case $k \ll n$. Note that in theorem Theorem~\ref{thm: generalize LB} the value $k$ is actually the number of vertices in the original lower bound graph, but this also upper bounds the size of the size of the solution.
    
    Our assumption regarding having a vertex with id $0$ is rather weak and both the lower bounds of \cite{KuhnMW16} and \cite{Censor-HillelKP17} still hold given this assumption. Note that when $\alpha=1$ we are guaranteed an exact solution in $G$. 
    This, in turn, means that taking the constructions where the number of vertices in the lower bound graph is equal to $k$ and then attaching $S_n$ will result in the following generalization of \cite{Censor-HillelKP17}:
    
    \begin{theorem}[Parametrized \cite{Censor-HillelKP17}]
		\label{thm:ckp small k}
		\ThmCKPParam
	\end{theorem}
	\begin{proof}
	The extension to the case where $k \ll n$ is direct from Theorem~\ref{thm: generalize LB}. Let us explain why the result of \cite{Censor-HillelKP17} also applies to $\kMVC$.
	We note that the result of \cite{Censor-HillelKP17} was for the problem of computing an exact $\MVC$, the reason we can restate it as a lower bound for $\kMVC{}$, is due to the fact that the lower bound uses a disjointness argument. Namely, edges are added to a basic graph construction representing two sets $A,B$. If the sets are disjoint the $\MVC$ has size $k_1$, otherwise $k_2 \neq k_1$. Solving $\MVC$ exactly allows us to distinguish between the two values, and solve the disjointness problems - the lower bound of \cite{Censor-HillelKP17} follows.
	
	Assume without loss of generality $k_1 < k_2$. Solving \kMVC{} on the construction of \cite{Censor-HillelKP17} with $k=k_1$ also allows us to distinguish between the two cases when the size of the \MVC{} equals $k_1$ (a vertex cover is found) or $k_2$ (all vertices reject). Thus, the lower bound of \cite{Censor-HillelKP17} also applies for \kMVC.
	\end{proof}
	
    As any cover of size at most $k$ implies an independent set of size at least $n-k$, the above also gives a lower bound on \kMaxIS{}, showing that it is hard to solve in the CONGEST model unless $k$ is almost as large as $n$:
    \begin{corollary}
    There exists $\underline{k}\in\mathbb N$ such that for any 
    $\underline{k}\le k\le 0.99n$, 
    any algorithm that solves \kMaxIS{} in the CONGEST model must take $\Omega(n^2 / \log^2 n)$ rounds, where $n = \size{V}$ can be arbitrarily larger than $k$.
    \end{corollary}
    
    Both our lower bound and that of \cite{KuhnMW16} use indistinguishability arguments to achieve approximation lower bounds. Namely, that there exists two graphs $G, G'$ (or a family of such graphs), such that in less than $t$ rounds the algorithm cannot tell them apart, resulting in the desired approximation error. This translates directly to parameterized approximation lower bounds, because even if $k$ is known the algorithm still cannot tell $G$ and $G'$ apart in less than $t$ rounds. That is, the argument is based on the topology learned by every node in the graph (which would be identical for $G, G'$), and thus the knowledge of $k$ will does not change the lower bound. We use Theorem~\ref{thm: generalize LB} together with the above to restate our lower bound and that of \cite{KuhnMW16}. Note that the lower bound of \cite{KuhnMW16} with respect to $\Delta$ does not translate to the parameterized setting.
    
    
	\begin{theorem}\label{thm:LB-main-param}
        \ThmLBMainParam
    \end{theorem}
    Note that by setting $\epsilon = 1/k$, the above implies an $\Omega(k)$ lower bound for computing an exact solution of any problem in \k{$\mathcal{P}$}.
    \begin{theorem}[Parametrized \cite{KuhnMW16}]
		\label{thm:kuhn small k}
		\ThmKuhnParam  
	\end{theorem}

\section{Upper Bound Warmup -- Parameterized \mbox{Diameter Approximation}}
In this section, we illustrate the concept of parameterized algorithms with the classic problem of diameter approximation. This procedure will also play an important role in all our algorithms.
%
Computing the exact diameter of a graph in the CONGEST model is costly. Specifically, it is known that computing a $(3/2-\epsilon)$-approximation of the diameter, or even distinguishing between diameter $3$ or $4$, requires $\widetilde O(n)$ CONGEST rounds~\cite{abboud2016near,bringmann2018note}.
Computing a $2$-approximation for the diameter is straightforward in $O(D)$ rounds, by finding the depth of any spanning tree. However, we wish to devise algorithms whose round complexity is bounded by some function $f(k)$, even if no solution of size $k$ exists.
Therefore, we now show that it is possible to compute a $2$-approximation for the \emph{parameterized version} of the diameter computation problem.
\begin{theorem}
\label{thm:diameter}
\ThmDiam
\end{theorem}
\begin{proof}
Our algorithm starts with $k$ rounds, such that in every round each vertex broadcasts the minimal identifier it has learned about (initially its own identifier). 
After this stage terminates, each vertex $v$ has learned the minimal identifier $x_v$ in its $k$-hop neighborhood.

Next follows $2k+1$ rounds such that in each round each vertex $v$ broadcasts $y_v$ and $z_v$, which are the minimal and maximal $x_u$ identifier it has seen so far.
When this ends, each vertex returns \emph{SMALL} if $y_v=z_v$ and \emph{LARGE} otherwise.
Clearly, the entire execution takes $O(k)$ rounds.

For correctness, observe that if the diameter is bounded by $k$ then all $x_v$'s are identical to the globally minimal identifier. Next, assume that the diameter is at least $2k+1$, and fix some vertex $v$. 
This means that there exist a vertex $u$ that whose distance is exactly $k+1$ with respect to $x_v$, and thus at most $2k+1$ from $v$. Since the first stage of the algorithm runs for $k$ rounds, we have that $x_u\neq x_v$. Therefore, after $k+1$ rounds of the second stage we have that $y_{x_v}\neq z_{x_v}$, and after additional $k$ rounds $y_v\neq z_v$ and thus $v$ outputs \emph{LARGE}. 
Finally, if the diameter is between $k+1$ and $2k$, then all vertices have the same $y_v$ and $z_v$ values and thus answer unanimously.
\end{proof}

\section{Parameterized Problems Upper Bounds}
\ifdefined\fullVersion
In this section, we discuss algorithms for the parameterized variant of many optimization problems.
\fi
\subsection{LOCAL Algorithms}
Our first result is for diameter lower bounded problem in the LOCAL model. We show that any minimization problem $\k{P} \in \DLB$ can be solved in $O(k)$ LOCAL rounds. 
To that end, we first use Theorem~\ref{thm:diameter} to check whether the diameter is at most $ck$, or at least $2ck$, where $c$ is a constant such that a diameter of $ck$ implies that no solution of size $k$ exists. If the diameter is larger than $2ck$, the algorithm terminates and reports that no $k$-sized solution exists. Otherwise, we collect the entire graph at a leader vertex $v$ which computes the optimal solution. If the solution is of size at most $k$, $v$ sends it to all vertices. If no solution of size $k$ exists, $v$ \mbox{notifies the other vertices.}

The above approach does not necessarily work for maximization problems as the existence of a $k$-sized solution does not imply a bounded diameter. Nevertheless, we now show that \kMaxM{} and \kMaxIS{} have $O(k)$ rounds algorithms.
For this, we first check whether the diameter is at most $2k$ or at least $4k$. If the diameter is small, we can still collect the graph and solve the problem locally. Otherwise, we use the fact that \emph{any} maximal matching or Independent Set in a graph with a diameter larger than $2k$ must be of size at least $k$. 
Since the maximal matching or independent set may be too large, we run just $k$ iterations of extending the solution.
For \kMaxM{}, at each iteration, any edge that neither of its endpoints is matched and that is a local minimum (with respect to the identifiers of its endpoints) joins the matching. We are guaranteed that the matching grows by at least a single edge at each round and thus after $k$ iterations the algorithm terminates. Similarly, for \kMaxIS{}, at each iteration, every vertex that neither of its endpoints is in the independent set and is a local minimum enters the set. 
We summarize this in the following theorem.
\begin{theorem}
\label{thm: local UB}
\ThmLocalUB
\end{theorem}

	\subsection{CONGEST algorithms for \kMVC{}}
	Our first algorithm is deterministic and aims to solve the exact variant of \kMVC{}. 
    Intuitively, it works in two phases; first, it checks that the diameter is $O(k)$, if not the algorithm rejects. Knowing that the diameter is bounded by $O(k)$, we proceed by calculating a solution \emph{assuming there exists a solution of size at most $k$}. If this assumption holds, we are guaranteed to find such a solution. 
    We run the above for just enough rounds to guarantee that if a $k$-sized solution exists we will find such.
    Finally, we check that the size of the solution returned by the algorithm is \mbox{indeed bounded by $k$.}
	
	We first show a procedure that solves the problem, if a solution of size $k$ exists. Note that if no such solution exists, this procedure may not terminate in time or compute a cover larger than $k$.
	\begin{lemma}\label{lem:promisekMVC}
	There exists a deterministic algorithm that if a $k$-sized cover exists: (1) terminates in $O(k^2)$ CONGEST rounds and (2) finds such a cover.
	\end{lemma}
	\begin{proof}
	Given that there exists a $k$-sized cover, the diameter of the graph is bounded by $2k$.
	Therefore, we can compute a unique leader and a BFS tree rooted at that leader in $O(k)$ rounds. 
	Our first observation is that every $v$ with a degree larger than $k$ must be in any $k$-sized cover. Thus, every such vertex immediately goes into the cover and gets removed together with all of its adjacent edges. If a vertex has degree $0$, it terminates (without entering the cover). Denote the remaining graph by $G'=(V',E')$. 
	
	For our analysis, let us fix some vertex cover   $X \subseteq V'$ of size $k$ and denote the remaining vertices by $A=V'\setminus X$. 
    We note that the set $A$ is an independent set. Thus, all edges in the graph are either between vertices in $X$ or between $A$ and $X$. We note that $\size{X} \leq k$, and now we aim to bound the number of remaining edges. 
    We now show that $\size{E'} \leq k^2$.
    
    
    As all vertices with degrees greater than $k$ have been added to the cover and removed, all remaining vertices have a degree of at most $k$. Because all remaining edges in the graph are of the form $(x,v)\in E', x\in X, v \in V'$ or $(u,v)\in E', u,v \in V'$, we may immediately bound the number of remaining edges, $\size{E'} \leq k^2$.
	
    
	We now can just learn the entire graph in time $O(\size{E'} + \size{D}) = O(\size{E'}) = O(k^2)$ using pipelining. The leader vertex computes the optimal cover for $G'$ and notifies all vertices in $G'$ whether they should join it.
	\end{proof}
	\begin{theorem}\label{thm:exactDeterministicVCAlg}
    There exists a deterministic algorithm for \kMVC{} that terminates in $O(k^2)$ rounds.
    \end{theorem}
	\begin{proof}
    Our algorithm first uses Theorem~\ref{thm:diameter} to estimate the diameter. 
	Specifically, we first apply it for $k'=2k$. If the vertices report LARGE, we follow the same approach as in the LOCAL algorithm and reject. Otherwise, we proceed knowing the diameter is bounded by $4k$. For the case the diameter is bounded we can compute a unique leader and a BFS tree rooted at that leader in $O(k)$ rounds.
	
	Let $c=O(1)$ such that the algorithm provided in Lemma~\ref{lem:promisekMVC} is guaranteed to terminate in $ck^2$ rounds for \emph{any} graph $G$ with a cover of size $k$.
	We run the algorithm for $ck^2$ rounds; if the procedure did not terminate, all vertices report that no $k$-sized cover exists.
    	
	Finally, we count the number of nodes in the cover using the BFS tree and verify that indeed the size of the solution in $G$ is bounded by $k$. If any node in the tree sees more than $k$ identifiers of vertices that joined the cover, it notifies all vertices that the solution is invalid and thus no $k$-sized solution exists.
	%
	\end{proof}

\paragraph{A Randomized Algorithm}
While we show a deterministic LOCAL algorithm for \kMVC{} that is optimal even if randomization is allowed, we have a gap of $\Theta(\min\set{k,\log k\log n})$ in our CONGEST round complexity. We now present a randomized algorithm with a $O\parentheses{k+\frac{k^2\log k}{\log n}}$ round complexity, thereby reducing the gap to $O(\log^2 k)$.
This is achieved by the observation that while node identifiers are of length $\Theta(\log n)$, we can replace each node identifier with an $O(\log k)$-bit \emph{fingerprint}. 
Specifically, since there are at most $k+k^2\le 2k^2$ vertices in $G'$ (after our reduction rule), we can use $(\mathfrak{b} = (c+4)\log k+1)$-bit fingerprints, for some $c>0$, and get that the probability of collision (that two vertices have the same fingerprint) is at most
$
{\binom{2k^2}{2}}2^{-\mathfrak{b}} < k^42^{-\mathfrak{b}+1}=k^{-c}.
$
Next, we run our deterministic algorithm, where each vertex considers its fingerprint as an identifier. Observe that since $|E'|\le k^2$ and each edge encoding now requires $O(\log k)$ bits (for $c=O(1)$), the overall amount of bits sent to the leader is $O(|E'|\log k)=O(k^2\log k)$. Since the diameter of the graph is $O(k)$, and $O(\log n)$ bits may be transmitted on every round on each edge, we use pipelining to get the round complexity below. 
Note that we only use fingerprints for the part of the algorithm which requires time quadratic in $k$. That is, checking the size of the diameter and validating the size of the solution are still done using the original identifiers.

\begin{theorem}\label{thm:exactRandomizedVCAlg}
For any $\delta = k^{-O(1)}$, there exists a randomized algorithm for \kMVC{} that terminates in $O\parentheses{k+\frac{k^2\log k}{\log n}}$ rounds, while being correct with probability at least $1-\delta$.
\end{theorem}
    
\paragraph{Approximations} 
\ifdefined\fullVersion
As we may add all nodes of degree more than $k$ to the cover, this bounds $\Delta$, the maximum degree in the remaining graph, by $k$. We can now apply the algorithm of \cite{ImprovedVC} which runs in $O(\log \Delta / \log \log \Delta + \log \epsilon^{-1} \log \Delta / \log^2 \log \Delta)$ and achieves a $(2+\epsilon)$-approximation. This immediately results in a deterministic $O(\log k / \log \log k + \log \epsilon^{-1} \log k / \log^2 \log k)$-round $(2+\epsilon)$-approximation algorithm in the CONGEST model.
Further, since there exists a cover of size $\OPT{}\le k$, setting $\epsilon=1/(k+1)$ implies that the resulting cover is of size $\floor{\OPT{}(2+\epsilon)} \le 2\OPT{} + \floor{k/(k+1)} = 2\OPT{}$. Thus, we conclude that our algorithm computes a $2$-approximation for the problem in $O(\log^2 k / \log^2 \log k)$ rounds. 
Unfortunately, while this succeeds if there indeed exists a solution of size $k$, validating the size of the solution takes $O(k)$ rounds. 

\else
After adding all nodes of degree more than $k$ to the cover, the maximum degree becomes bounded by $k$. This allows us to use existing algorithms that have runtime of $O(\log \Delta / \log \log \Delta)$ for computing a $2+\eps$ approximation \emph{if such exists}. However, validating the size of the solution takes $O(k)$ rounds. 
\fi
%
We now expand the discussion and propose an algorithm that computes a $(2-\eps)$-approximation. For $\epsilon=o(1)$, this gives a better round complexity than \mbox{our exact algorithm,} while for $\eps=O(1/\sqrt{k})$ it improves the approximation ratio of the above $2$-approximation while still terminating in $O(k)$ rounds.
In Section~\ref{sec:classicProblems}, we use this algorithm to derive the first $2-\eps$ algorithm for the (non-parametric) \MVC{} problem that operate in $o(n^2)$ rounds.
\begin{theorem}\label{thm:2-epsPromiseApproximation}
$\forall \eps\in[1/k,1]$, there exists a deterministic CONGEST algorithm for \kMVC{} that computes a $(2-\epsilon)$-approximation in $O\parentheses{k+(k\epsilon)^2}$ rounds.
For any $\delta = (k\epsilon)^{-O(1)}$, there also exists a randomized algorithm that terminates in $O\parentheses{k+\frac{(k\epsilon)^2\log (k\epsilon)}{\log n}}$ rounds and errs \mbox{with probability 
$\le\delta$.}
\end{theorem}
\begin{proof}
We utilize the framework of Fidelity Preserving Transformations~\cite{FELLOWS201830}. 
Intuitively, if there exists a vertex cover of size $k$ in the original graph, and we remove two vertices $u,v$ that share an edge, then the new graph must have a cover of size at most $k-1$. This allows us to reduce the parameter at the cost of introducing error (we add both $u$ and $v$ to the cover, while an optimal solution may only include one of them). This process is called \emph{$(1,1)$-reduction step} as it reduces the parameter by $1$ and increases the size of the cover (compared with an optimal solution) by at most $1$.
Roughly speaking, we repeat the process until the parameter reduces to $k\epsilon$, at which point we run an exact algorithm on the remaining graph.

In~\cite{FELLOWS201830}, it is proved that for any $\alpha\le 2$, repeating a $(1,1)$-reduction step until the parameter reduces to $k(2-\alpha)$ allows one to compute an $\alpha$-approximation by finding an exact solution to the resulting subgraph and adding all vertices that have an edge that was reduced in the process. For our purpose, we set $\alpha=2-\epsilon$; thus, the exact algorithms only need to find a cover of size $k\epsilon$.

Our algorithm begins by checking the diameter is $O(k)$ and finding a leader vertex $v$. This is doable in $O(k)$ rounds as having a vertex cover   of size $k$ guarantees that the diameter is $O(k)$. 
We proceed with applying the $(1,1)$-reduction steps.
To that end, we compute a maximal matching $M$ and send it to $v$, which requires $O(k)$ rounds. 
If $|M|\le k(1-\epsilon)$, $v$ instructs all matched vertices to enter the cover, and the algorithm terminates with a solution of size at most $2k(1-\epsilon) < k(2-\epsilon)$, as needed.
If $|M|> k(1-\epsilon)$, the leader selects an arbitrary submatching $M'\subset M$ of size $k(1-\epsilon)$ and the reduction rules are simultaneously applied for every $e\in M'$.
The remaining graph has a cover of size at most $k\epsilon$, at which point we apply the exact algorithms. 
Finally, we validate the size of the solution as in the above algorithms.
By theorems~\ref{thm:exactDeterministicVCAlg} and~\ref{thm:exactRandomizedVCAlg}, we get the stated runtime and establish the correctness of our algorithms.
\end{proof}

\subsection{CONGEST algorithms for \kMaxM{}}

    %
	Similar to \kMVC{}, we first design a deterministic algorithm for \kMaxM{} that terminates in $O(k^2)$ rounds.
	Our algorithm first uses Theorem~\ref{thm:diameter} to estimate the diameter. 
	Specifically, we first apply it for $k'=2k$. 
	If the vertices report LARGE (which means that the diameter is at least $2k+1$), we follow the same approach as in the LOCAL algorithm and make $k$ iterations of the maximal matching algorithm.  The large diameter ensures that any maximal matching is of size at least $k$ and thus we can terminate after these $k$ iterations.
	
	In case the output of the diameter approximation was SMALL, we know that the diameter is at most $4k$. This allows us to compute a leader in $O(k)$ rounds. In this case, we also run the maximal matching algorithm for $k$ iterations, but now this might not be sufficient. 
	That is, the size of the maximal matching might be smaller than $k$, while the size of the maximum matching is $k$ or larger. To address this issue, we augment \cite{lovasz2009matching} the matching until it is of size at least $k$, or we cannot augment the matching any further, at which point we conclude that no matching of size $k$ exists.
    
    
    
    
    \paragraph{Augmenting a maximal matching} We prove that given some \emph{maximal} matching $M, \size{M}<k$ we can find a matching of size at least $k$, or determine that such does not exist, in $O(k^2)$ rounds. 
    If we reached a maximum matching of a size smaller than $k$, the vertices report that no solution to the \kMaxM{} instance exists.
	
	We first note the subgraph $G_M$ induced by the nodes in $V_M = \set{v\in V \mid \exists u\in V, (v,u) \in M}$ has at most $2k$ nodes and $O(k^2)$ edges. Also note that because the diameter is bounded, a leader can be elected in $O(k)$ rounds, when the communication is done over $G$. \mbox{Now we state our algorithm.}
	
	Every node $v \in V_M$ picks $2$ unmatched neighbors from $V \setminus V_M$ and sends them to the leader node. If it has less than 2, it sends $N(v)\cap (V \setminus V_M)$. The leader node then decides upon an augmenting path and augments accordingly. This process repeats at most $k$ times. If the matching is sufficiently large, we output the matching; otherwise, all nodes output that no solution exists.
	
	As for correctness we first note that every augmenting path in $G$ has two endpoints in $V \setminus V_M$ and the rest of the nodes are in $V_M$. This implies that apart from the endpoint edges, all edges of the augmenting path are in $G_M$. We prove the following lemma:
	
	\begin{lemma}
		If there exists an augmenting path in $G$, then the algorithm finds such.
	\end{lemma}
	\begin{proof}
	    Fix some augmenting path in $G$, and denote by $u_s, u_f \in V \setminus V_M$, its endpoints. Let $v_s, v_f \in V_M$ be the nodes adjacent to $u_s, u_f$ on the path, and let $S(v_s), S(v_f)$ be the sets of unmatched neighbors chosen by the nodes in the algorithm. Note that these sets cannot be empty. If $\size{S(v_s)}=\size{S(v_f)}=1$, then the nodes chosen are exactly those of the augmenting path and we may augment. 
	    Otherwise, 
	    we have $\size{S(v_s) \cup S(v_f)}\geq2$, in which case we again found an augmenting path. 

	As for the running time, each augmentation rounds takes $O(k)$ rounds using pipelining, while due to the bound on the matching size there can be at most $O(k)$ such augmentations. Thus the total running time is bounded by $O(k^2)$.
	\end{proof}
	Similarly to \kMVC{}, we show that by replacing the node identifiers with random $O(\log k)$-sized fingerprints we get a faster randomized algorithm.
	\begin{theorem}
	\label{thm: MaxM Rand}
	For any $\delta = k^{-O(1)}$, there exists a randomized algorithm for \kMaxM{} that terminates in $O\parentheses{k+\frac{k^2\log k}{\log n}}$ rounds, while being correct with probability at least $1-\delta$.
	\end{theorem}
	\begin{proof}
	First, assume that a $k$-sized matching exists, and let us consider some execution of the deterministic algorithm above. 
	Recall that it starts from a maximal matching $M$ and extends it by finding augmenting paths until it becomes a maximum matching   or reaches size $k$. We shorten node identifiers in the part of the algorithm which requires time quadratic in $k$.
	Let $V'$ be the set of all nodes that were matched during some point of the execution of the algorithm.
	Since $\size{V_M}\le 2k$ and at most $2$ new vertices are matched during each of the augmentations, we have that $\size{V'}\le 4k$.
	
	For our randomized algorithm, after computing the maximal matching $M$, each node generates a random $(\mathfrak{b} = (c+2)\log k+4)$-bit fingerprints, for some $c>0$.
	We then continue with augmentations, same as is the deterministic algorithm, but each node uses its fingerprint instead of its identifier.
	We have that if all vertices in $V'$ have distinct fingerprints, then the algorithm is successful in finding a $k$ sized matching. The probability of collision by any two vertices in $V'$ is 
	%
$
{\binom{4k}{2}}2^{-\mathfrak{b}} < k^2 2^{-\mathfrak{b}+4}=k^{-c}.
$ Finally, since the overall number of bits sent to the leader is $O(k^2\log k)$, on a graph with diameter $O(k)$, our runtime follows.
	\end{proof}
\noindent\quad{}Similarly to Theorem~\ref{thm:2-epsPromiseApproximation}, we show that \kMaxM{} admits a faster $(2-\eps)$-approximation \mbox{for $\eps=o(1)$.}
\begin{theorem}\label{thm:2-epsMaxM}
$\forall\eps\in[1/k,1]$, there exists a deterministic algorithm for \kMaxM{} that computes a $(2-\epsilon)$-approximation in $O\parentheses{k+(k\epsilon)^2}$ CONGEST rounds.
For any $\delta = (k\eps)^{-O(1)}$, there also exists a randomized algorithm that terminates in $O\parentheses{k+\frac{(k\epsilon)^2\log (k\epsilon)}{\log n}}$ rounds and errs \mbox{with probability $\le \delta$.}
\end{theorem}	
\begin{proof}
For \kMaxM, the algorithm is similar than that of \kMVC. 
Since adding any edge $e\in E$ to a matching introduces an additive error of at most $1$ while decreasing the parameter by $1$, this is a $(1,1)$-reduction rule. Once again, we first check that the diameter is $O(k)$ by invoking Theorem~\ref{thm:diameter} for $k'=2k$. We again pick a leader $v$ and send him a maximal matching $M$. If $|M|\ge k/(2-\eps)$, the algorithm terminates. Otherwise, $v$ picks a submatching $M'\subseteq M$ of size $\max\set{k/2-2k\eps,0}$ (this is possible as $|M|\ge k/2$) and messages all matched vertices. The remaining graph (after all vertices matched in $M'$ are removed) has a matching of size at most 
$$
2(\size{M}-\size{M'})\le 2(k/(2-\eps)-(k/2-2k\eps))
\le k/(1-\eps)-k+4k\eps
\le k(1+2\eps)-k+4k\eps = 6k\eps.
$$
This is because nodes outside of $M$ are an independent set and each node that is matched in $M$ but not $M'$ can only add one edge to the matching. 
so we can find the maximum matching   in $O(k+(k\epsilon)^2)$ rounds deterministically or $O\parentheses{k+\frac{(k\epsilon)^2\log (k\epsilon)}{\log n}}$ using our randomized algorithm (Theorem~\ref{thm: MaxM Rand}).
Since we applied $k/2-2k\eps$ reduction rules, if a matching of size $k$ exists then the resulting matching must be at least of size $k-(k/2-2k\eps)=k\frac{1+2\eps}{2}\ge\frac{k}{2-\eps}$. Finally, the leader verifies that the obtained maximal matching is at least \mbox{of size $k$ and informs all vertices.}
\end{proof}

\section{Faster Algorithms for Non-Parameterized Problems}\label{sec:classicProblems}
In this section, we discuss how parameterized algorithms can be used to obtain improved round complexity for classical distributed problems.
That is, here we \emph{do not} assume that we are given a parameter $k$ and design algorithms that address the original problem variants. Our goal is to devise algorithms that solve the problems fast \emph{if} the optimal solution is small. 

For a problem $P$, denote by $T_P(k,n)$ the runtime of an algorithm for $k$-$P$ on an $n$-nodes graph. Next, denote by $\OPT{}$ the value of the optimal solution to $P$ on a given input graph $G$ on $n$ nodes.
Then we can solve the problem in a time that depends on $\OPT{}$, regardless of $n$.
Here, the optimal approach depends on $T_P(k,n)$. For example, a linear scan would yield a $O\parentheses{\sum_{k=0}^{\OPT{}}T_P(k,n)}$ rounds solution, which may be good if $T_P(k,n)$ grows fast in $k$ (e.g., this is optimal for $T_P(k,n)=2^{\Omega(k)}$).
We proceed with search procedures that are more suitable for our $T_P(k,n)=k^{O(1)}$ algorithms.

\begin{theorem}\label{thm:fastClassicalVC}
For any $\epsilon\in(0,1)$, there exist deterministic algorithms that compute a $(1+\eps)$-approximation of \MVC{} or \MaxM{} and terminate in $O\parentheses{\min\set{n^2, \OPT{}^2\log\epsilon^{-1}}}$ CONGEST rounds. 
\end{theorem}
\begin{proof}
We start by running our \kMVC{} (respectively, \kMaxM{}) algorithm for $k=2^0,2^1,\ldots$ until it returns a valid solution (respectively, reports that no $k$-sized matching exists). At this point, we have a bound $\overline k$ such that $\overline k/2 < \OPT{} \le \overline k$.
If $\overline k = \Omega(n^2)$, we run the naive $O(n^2)$ rounds solution and terminate.
Next, we perform $t=\ceil{\log\eps^{-1}+1}$ iterations of binary search for the value of \OPT{}. That is, we start with $k=3\overline k/4$ and proceed either to $7\overline k/8$ or $5\overline k/8$, depending on whether a cover of size $3\overline k/4$ exists.
After these $t$ iterations, there exists some $\alpha\in[1/2,1)$ such that a cover / matching of size $\alpha \overline k$ does not exist, but a cover / matching of size $(\alpha+2^{-t})\overline k$ does. By returning the $(\alpha+2^{-t})\overline k$-sized solution we obtain an approximation ratio of 
$$
\frac{(\alpha+2^{-t})\overline k}{\alpha \overline k} = 1+2^{-t}/\alpha \le 1+\eps.
$$
The search for $\overline k$ requires $O(\OPT{}^2)$ rounds while the binary search takes $O(\OPT{}^2 \log\epsilon^{-1})$ time, which concludes the proof.
\end{proof}
We note that without fixing $\epsilon$ in advance, one could complete the binary search to find an optimal solution, as we state in the following.
\begin{corollary}
There exist deterministic exact algorithms for \MVC{} and \MaxM{} that terminate in $O\parentheses{\min\parentheses{n^2, \OPT{}^2\log \OPT{}}}$ CONGEST rounds. 
\end{corollary}
We proceed with algorithms for approximating \MVC{} and \MaxM{}. Observe that for $\eps=o(1)$ this gives faster algorithms than the above.
\begin{theorem}
For any $\eps\in (0,1/2]$, there exist deterministic $(2-\epsilon)$-approximation algorithms for \MVC{} and \MaxM{} that terminate in $O\parentheses{ \log\eps^{-1}\parentheses{\OPT{}+(\eps \OPT{})^2}}$ CONGEST rounds.
\end{theorem}
\begin{proof}
Similarly to the algorithm of Theorem~\ref{thm:fastClassicalVC}, we start by scanning for $k=2^0,\ldots$, but we run the $(2-2\epsilon)$-approximation of \kMVC{} or \kMaxM{} (see Theorems~\ref{thm:2-epsPromiseApproximation} and~\ref{thm:2-epsMaxM}) rather than the exact solution.
When we first find a vertex cover   (or learn that no $k$-sized matching exists), we get a number $\overline k$ such that $\overline k/2< \OPT{} \le \overline k(2-2\epsilon)$. At this point we run $t=\ceil{\log\eps^{-1}+2}$ iterations of binary search. This gives us a quantity $\alpha\in[1/2,1)$ such that $\alpha\overline k< \OPT{} \le (\alpha+2^{-t})\overline k(2-2\epsilon)$. By returning the $(\alpha+2^{-t})\overline k(2-2\epsilon)$ sized solution we get an approximation ratio of $$
\frac{(\alpha+2^{-t})\overline k(2-2\epsilon)}{\alpha\overline k}
\le 2 - 2\epsilon + 2^{1-t}/\alpha \le 2-\epsilon.
$$
As our deterministic round complexity of computing $(2-2\eps)$-approximations for \kMVC{} and \kMaxM{} is $O(k+(k\epsilon)^2)$, our runtime is as stated.
\end{proof}
By selecting $\eps=1/2\sqrt{\overline{k}}$ after the initial phase of the above algorithm, we conclude the following result. We note that these are the first deterministic algorithms that terminate in $o(n^2)$ rounds while having an approximation ratio strictly lower than $2$. Further, for \MVC{}, no such randomized algorithm is known.
\begin{theorem}
There exist $O(\OPT{}\log \OPT{})$ rounds deterministic CONGEST algorithms that compute a $(2-1/\sqrt{\OPT{}})$-approximation of \MVC{} and \MaxM.
\end{theorem}
\newcommand{\tmvc}[1][\alpha]{\ensuremath{T^{\MVC}_{#1}}}
\newcommand{\tmaxm}[1][\alpha]{\ensuremath{T^{\MaxM}_{#1}}}
\newcommand{\tP}[1][\alpha]{\ensuremath{T^{P}_{#1}}}
We now proceed to computing a $2$ and $(2+\epsilon)$-approximations for \MVC{} and \MaxM{}. In the following, for some $\alpha\ge 1$, let \tmvc{} (respectively, \tmaxm{}) denote the time required for deterministically computing an $\alpha$-approximation for \MVC{} (respectively, \MaxM{}) for any graph of $n$ nodes.
It is known that $\tmvc[2+\epsilon]=O\parentheses{\frac{\log\Delta}{\log\log \Delta}+\frac{\log \eps^{-1}\log\Delta}{\log^2\log \Delta}}$ and $\tmvc[2]=O\parentheses{\frac{\log n\log\Delta}{\log^2\log \Delta}}$ for graphs with degrees bounded by $\Delta$~\cite{ImprovedVC}. It is also known that $\tmaxm[2]=O(\Delta+\log^* n)$~\cite{Bar-YehudaCGS17}.
The following shows that, for $P\in\set{\MVC,\MaxM}$, if $\OPT{}\log \OPT{}=o\parentheses{\tP[2]}$ we can obtain a $2$ approximation for $P$, and if $\OPT{}=o\parentheses{\tP[2+\epsilon]}$ we get a faster algorithm for $(2+\eps)$-approximation.
Additionally, for any $\eps=\OPT{}^{-o(1)}$ this gives a speedup with respect to the above algorithms. Specifically, for constant $\eps$ we get a \mbox{$(2+\eps)$-approximation in $O\parentheses{\min\set{\OPT{},\tP[2+\eps]}}$ rounds.}
\begin{theorem}
Let $P\in\set{\MVC,\MaxM}$. There exists a deterministic $2$-approximation for $P$ that terminates in $O\parentheses{\min\set{\OPT{}\log \OPT{}},\tP[2]}$ rounds in the CONGEST. There also exists a deterministic $(2+\eps)$-approximation that terminates in $O\parentheses{\min\set{\OPT{}\log\eps^{-1}, \tP[2+\epsilon]}}$ rounds.
\end{theorem}
\begin{proof}
Let $\alpha\in\set{2,2+\epsilon}$ be the desired approximation ratio.
We run the $2$-approximation algorithm of \k{P} for $k=2^0,2^1,\ldots$, until we get a valid cover (alternatively, learn that no $k$-sized matching exist) or reach $k\approx {\tP[\alpha]}$. 
If we reached $k\approx {\tP[\alpha]}$, we run the $\tP[\alpha]$ rounds algorithm and terminate.
Otherwise, this gives a number $\overline k$ such that $\overline k/2< \OPT{} \le 2\overline k$. 
If $\alpha=2+\eps$, we perform $t=\ceil{\log\eps^{-1}+2}$ iterations of binary search, we reduce the range in a similar way to the above theorems and obtain a $(2+\eps)$-approximation in the stated number of rounds.
If $\alpha=2$, we complete the binary search and obtain a $2$-approximation in $O(\OPT{}\log \OPT{})$ rounds.
\end{proof}

\paragraph{Acknowledgements} The authors thank Keren Censor-Hillel, Seri Khoury, and Ariel Kulik for helpful discussion and comments.

\ifdefined\shortVersion
\newpage
\appendix
	\section{Reducing $\MVC$ to $\MFVS, \MFES$}\label{app:MVC2MFVS}
	We present reductions from $\MVC$ to $\MFVS$ and $\MFES$. Despite being folklore, we present them here for completeness and to point out that they work in the distributed setting without issues (no congestion is incurred).
	
	Let $G(V,E)$ be an instance of \MVC{} and define by $G_1=(V\cup E, \set{\set{v,u},\set{v,e},\set{u,e}\mid e=\set{u,v}\in E})$ the graph in which we replace each edge $e=\set{u,v}$ with a triangle $\set{u,e},\set{v,e},\set{u,v}$. 
	Further, let $\Vec G_2=(\set{v_{in},v_{out}\mid v\in V}, \set{(v_{in},v_{out})\mid v\in V}
	\cup\set{(u_{out},v_{in}),(v_{out},u_{in})\mid \set{u,v}\in E})$ be the digraph in which each $v\in V$ is replaced with two vertices $v_{in},v_{out}$ connected by an arc and each edge $\set{u,v}$ is replaced by the arcs $(u_{out},v_{in}),(v_{out},u_{in})$.
	We prove the following lemma:
	\begin{restatable}{lemma}{testname}
		There is an $\MVC$ of size $k$ in $G$ if and only there is an $\MFVS$ of size $k$ in $G_1$. Further, There is an $\MVC$ of size $k$ in $G$ if and only there is an $\MFES$ of size $k$ in $G_2$
	\end{restatable}
	\begin{proof}
	Assume that there is a cover $U\subseteq V$ of $G$, then $U$ is an FVS for $G_1$ and $\set{(v_{in},v_{out})\mid v\in U}$ is an FES for $G_2$. Similarly, let $F_1$ be an FVS for $G_1$ and let $F_2$ be a $FES$ for $G_2$. Denote by $F_1'$ the vertex set we get by arbitrarily replacing each $\set{u,v}\in F_1$ by one of its endpoint (e.g., $u$); then $F_1'$ is a cover of $G$. 
	Similarly, the set $F_2'=\set{v\mid ((v_{in},v_{out})\in F_2) \vee (\exists u:(v_{out},u_{in})\in F_2)}$ is a cover for $G$ (of size at most $|F_2|$). This concludes the proof.
	\end{proof}

	\paragraph{Distributed implementation} We note that the above reductions can be simulated in the distributed setting without incurring congestion. In \MFVS{}, $\set{u,v}$ can be simulated by either $u$ or $v$, and in the second reduction a node $v\in V$ can simulate $v_{in}, v_{out}$. 
	Finally, we state \mbox{our main theorem}:
	\begin{theorem-repeat}{thm: mvc reduction}
\ThmDiam
\end{theorem-repeat}

\section{Proof of Theorem~\ref{thm:LB-main}}\label{app:minimax}
We use Yao's Minimax Principle~\cite{minimax} to extend to lower bound to randomized algorithms and prove Theorem~\ref{thm:LB-main}.
\begin{theorem-repeat}{thm:LB-main}
\ThmLBMain
\end{theorem-repeat}
\renewcommand*{\proofname}{Proof of Theorem~\ref{thm:LB-main}}
\begin{proof}
In the proof of Lemma~\ref{lem:detApproxLB}, we showed that there is an identifier assignment that forces any deterministic algorithm to sub-optimally solve each path. For randomized algorithms, this does not necessarily work. However, we show that by applying Yao's Minimax Principle~\cite{minimax} we can get similar bounds to Monte Carlo algorithms. Intuitively, we consider an input distribution that randomly selects for each path whether to use the original ordering or whether to reverse $A_i$ or $B_i$ (this depends on the problem at hand, as in Lemma~\ref{lem:detApproxLB}). In essence, this gives the algorithm a chance of at most half of finding the optimal solution to each math.

Formally, fix an initial vertex identifier assignment $V\to\mathcal I$. 
For the problems of \MVC, \MaxM{}  and \MaxIS (respectively, \MDS and \MEDS) consider the $2^r$ inputs obtained by reversing any subset of $\set{A_i}_{i\in[r]}$ (respectively, $\set{B_i}_{i\in[r]}$). Next, consider the uniform distribution that gives each of these input probability $2^{-r}$. 
A similar argument to that of Lemma~\ref{lem:detApproxLB} shows that the probability that an algorithm computes the optimal solution to each path $i\in[r]$ is at most half.
We now apply a simplified version of the Chernoff Bound that states that for any $\gamma\in(0,1)$, $r\in \mathbb N^+$, a binomial random variable $X\sim \text{Bin}(r,1/2)$ satisfies $\Pr\brackets{X\ge r/2(1+\gamma)}\le 2EXP(-r\gamma^2/6)$. Plugging $\gamma=\sqrt{-6\logp{(1-2\delta)/2}/r}$ we get 
\begin{align*}
\Pr\brackets{X\ge t/2 + \sqrt{-6t\logp{(1-2\delta)/2}/4}}\le 2EXP(\logp{(1-2\delta)/2}
)=1-2\delta.
\end{align*}
That is, we get that any deterministic algorithm that is correct with probability $1-2\delta$ on a random input chosen according to $p$ cannot find the optimal path in more than $r/2 + \sqrt{-6r\logp{(1-2\delta)/2}/4}$ of the paths.
Next, we restrict the value of the error probability to $\delta \le 1/2-EXP(-2r/9)=1/2-EXP(-\Omega(n\epsilon))$, which guarantees that 
$$r/2 + \sqrt{-6r\logp{(1-2\delta)/2}/4} \le r/2 + \sqrt{-6r(-2r/9)/4}=5r/6.
$$
This means that in at least $r/6$ of the paths the algorithm fails to find the optimal solution and adds at least an additive error of one.
Therefore, since the optimal solution is of size $\Theta(n)$, the approximation ratio obtained by the deterministic algorithm is 
\begin{align*}
\alpha=1 + \Omega\parentheses{\frac{r}{n}}
= 1 + \Omega\parentheses{\frac{r}{r\ell}}  = 1+\Omega(1/x).
\end{align*}
Thus, by the Minmax principle we have that any Monte Carlo algorithm with less than $x$ rounds also have an approximation $1+\Omega(1/x)$. Using $x=\Theta(\epsilon^{-1})$, we established the \mbox{correctness of theorem.\qedhere}
\end{proof}
\renewcommand*{\proofname}{Proof}

\section{Formal Problem Definitions}\label{app:problems}
Let $G=(V,E)$ (Directed graph for \MFES) denote the target graph. 
We consider $U\subseteq V$ to be a feasible solution to the following vertex-set problems if:
\begin{itemize}
\item \problem{Minimum Vertex Cover (\MVC)}
$\forall e\in E: e\cap U\neq\emptyset$.
\item \problem{Maximum Independent Set (\MaxIS)}
$\forall u,v\in U: \set{u,v}\notin E$.
\item \problem{Minimum Dominating Set (\MDS)}
$\forall v\in V, \exists u\in U: \set{u,v}\in E$.
\item \problem{Minimum Feedback Vertex Set (\MFVS)}
$G[V\setminus U]$ is acyclic.
\end{itemize}
Next, we call $Y\subseteq E$ a feasible solution to the following edge-set problems if:
\begin{itemize}
\item \problem{Maximum Matching (\MaxM)}
$\forall v,u,w\in V: \set{v,u}\in Y\implies \set{u,w}\notin Y$.
\item \problem{Minimum Edge Dominating Set (\MEDS)} 
\mbox{$\forall e\in E, \exists e'\in Y: e\cap e' \neq \emptyset$.}
\item \problem{Minimum Feedback Edge Set (\MFES)} 
\mbox{$(V,E\setminus Y)$ is acyclic.}
\end{itemize}

\section{Optimal solutions for the problems on the graph $G_{r,2x+3}$}\label{app:optSolutions}
Here, we provide the exact characterization of the optimal solutions of the problems; recall that $(x\mod 6)=0$.
\begin{align}
&\OPT{}_{MVC}=\set{v_{i,2j}\mid i\in[r], j\in[x+2]}, 
\\&\OPT{}_{MM}=\set{\set{v_{i,2j},v_{i,2j+1}}\mid i\in[r], j\in[x+2]}, 
\\&\OPT{}_{MaxIS}=\set{v_0}\cup \set{v_{i,2j+1}\mid i\in[r], j\in[x+2]}, 
\\&\OPT{}_{MDS}=\set{v_{i,3j}\mid i\in[r], j\in\brackets{{\frac{2x}{3}+1}}}, 
\\&\OPT{}_{MEDS}=\set{\set{v_{0},v_{i,0}}\mid i\in[r]}\cup \set{\set{v_{i,3j+2},v_{i,3j+3}}\mid i\in[r], j\in\brackets{{\frac{2x}{3}}}}.\footnotemark{}
\end{align}
\footnotetext{Each edge of the form $\set{v_0,v_{i,0}}$ (for $i\in[r]$) may be replaced by the $\set{v_{i,0},{v_{i,1}}}$ edge.}
\fi

\bibliographystyle{alpha}
	\bibliography{paper}

\end{document}